\documentclass[12pt]{iopart}

 \usepackage{transparent}
\usepackage{floatrow}
\usepackage{cite}
\usepackage[colorlinks=true,breaklinks,
]{hyperref} 
\usepackage[hyphenbreaks]{breakurl} 
\usepackage{xcolor}
\definecolor{internal}{rgb}{0.6,0,0.2} 
\definecolor{citationC}{rgb}{0,0,0.6} 
\definecolor{c3}{rgb}{0.3,0,0.9} 
\hypersetup{
    linkcolor= {internal}, 
    citecolor={citationC}, 
    urlcolor={c3} 
}
\usepackage{color}
\usepackage{booktabs}
\usepackage{float}
\usepackage{amssymb}
\usepackage{amsthm}
\usepackage{graphicx}
\usepackage{subfig}
\usepackage[space]{grffile}
\usepackage{tensor}
\usepackage{amsfonts}
\usepackage{bm}
\usepackage[utf8]{inputenc}
\usepackage{slashed}
\usepackage{hyperref}
\usepackage[toc,page]{appendix}	
\usepackage{lipsum} 
\usepackage{comment}
\usepackage[percent]{overpic}
\usepackage{subfig}
\usepackage{tikz}
\usetikzlibrary{arrows.meta}
\usetikzlibrary{bending}
\usetikzlibrary{decorations.text}

\providecommand{\abs}[1]{\vert#1\vert}
\newtheorem{prop}{Proposition}
\newtheorem{thm}{Theorem}[section]
\newtheorem{lem}[thm]{Lemma}

\newtheorem{dfn}{Definition}
\def\beq{\begin{eqnarray}}
\def\eeq{\end{eqnarray}}

\definecolor{axisColor}{rgb}{0, 0, 0.5}

\begin{document}

\title[]{Approaching Wonderland
}

\author{Ben David Normann, Sigbj\o rn Hervik}
\date{today}

\address{Faculty of Science and Technology, University of Stavanger, 4036 Stavanger, Norway\\
}
\ead{ben.d.normann@uis.no, sigbjorn.hervik@uis.no}
\vspace{10pt}
\today

\begin{abstract}
Continuing previous work, we show the existence of stable, anisotropic future attractors in Bianchi invariant sets with a $p$-form field ($p\,\in\,\{1,3\}$) and a perfect fluid. In particular,  we consider the not previously investigated Bianchi invariant sets $\mathcal{B}$(II), $\mathcal{B}$(IV), $\mathcal{B}$(VII$_0$) and $\mathcal{B}$(VII$_{h})$ and determine their asymptotic behaviour. We find that the isolated equilibrium set Wonderland is a future attractor on all of its existence ($2/3<\,\gamma\,<2$) in all these sets except in $\mathcal{B}$(II), where the peculiar equilibrium sets Edge and Rope show up, taking over the stability for certain values of $\gamma$. In addition, in $\mathcal{B}$(IV) and $\mathcal{B}$(VII$_h$) plane gravitational wave solutions (with a non-zero $p$-form) serve as attractors whenever  $2/3<\,\gamma\,<2$. 
\end{abstract}
%
\vspace{2pc}

\section{Introduction}
\label{Intro}
\subsection{Context and background}
In a previous paper \cite{normann18} general equations for a perfect fluid and a homogeneous, sourceless $j$-form field (where $j=1,3$) in a cosmological context with general relativity were written down in an orthonormal frame. For more on the orthonormal frame approach, refer to \cite{elst97}. The underlying $(j-1)\,$\,-\,gauge field was not required to be homogeneous. As a result one may view the work as a study of an inhomogeneous, massless scalar gauge field with a homogeneous gradient. In the paper, we explicitly considered the Bianchi invariant sets $\mathcal{B}$(I) and $\mathcal{B}$(V) and provided a dynamical systems analysis of the cosmological evolution of such universes. In this paper we continue the work and consider some of the other Bianchi models of solvable type.

The $p$-form action obeys the weak energy condition. From a Hamiltonian point of view, it is bounded from below, as shown explicitly in \cite[Sec. 2.3]{thorsrud18}. Although the $j$-form field cannot sustain an accelerated state of expansion, it may still play an important role in the early universe, where attention is devoted to all sorts of isotropy-breaking fields \cite{hervik11,ford89,ackerman07,golovnev08,watanabe09,maleknejad13,ito15, almeida19,cicciarella19}. The goal is often to address the $\Lambda$CDM anomalies\cite{bennett11}. The $p$-form field is a simple way to incorporate an isotropy-breaking field in a general manner. Crucially, one should note that it is not the anisotropic curvature (of the Bianchi models) itself, but rather the shear in them, that produces large-scale CMB anomalies. Hence, a study of isotropy violating fields in the full set of homogeneous cosmologies is clearly of interest. In particular because of the existing shear-free solutions, such as the class of FLRW solutions and extensions to it. These matters are thoroughly investigated in \cite{thorsrud18,thorsrud20}.

The line of work contained in this paper and the previous, may also be placed in a broader context, where the goal is to understand how sensitive the evolution of the Universe is to initial conditions. A chief outcome would be to gain a clear understanding of how the observable Universe is so flat and isotropic. To this end non-tilted perfect fluids have already been studied in anisotropic backgrounds (cf. \cite{dynSys} and references therein), as well as tilted perfect fluids \cite{barrow03, hervik04a, hervik04b, coley05, hervik05, hervik07, hervik08, hervik10, shogin14, shogin15,barrow12} and fluids with vorticity \cite{hervik06a, hervik06b, coley08}. Naturally the relation to different inflationary scenarios has been discussed (e.g. \cite{barrow06, hervik11}) and the connection to observations has been investigated to some degree, for instance in \,\cite{jaffe06}. Source-free electromagnetism has also been investigated to some degree~\cite{leblanc97,yamamoto11}. Also, the e-book by A.~A.~Coley \cite{coleyDynSys} provides a comprehensible overview of a large range of studies with a variety of matter content.

The $p$-form action with $p\,\in\,\{1,3\}$ seems to have gone largely unnoticed in the cosmological literature so far, although it has gained some interest recently. One of the authors of our previous work has done work on shear-free cosmologies\cite{thorsrud18} with a $p$-form action. The reader is also referred to the short notice \cite{gruzinov04} and the more recent works \cite{almeida20,almeida19}. In this paper, we consider a $p$-form action (with $p\,\in\,\{1,3\}$) alongside a non-tilted perfect fluid.

The rest of the paper is structured as follows. In this section we write down the general set of equations and summarize some previously established results.\\
\textbf{ In Section \ref{Sec:Summary} we summarize the main results of the study}. In Section \ref{Sec:constraints} we give a general discussion of constraints, before we in Section \ref{Sec:SysEqIVEtc} focus attention on the particular sets {$\mathcal{B}$(II)}, $\mathcal{B}$(IV), $\mathcal{B}$(VII$_0$) and $\mathcal{B}$(VII$_{h})$. Of these sets $\mathcal{B}$(VII$_0$) and $\mathcal{B}$(VII$_{h})$ are especially interesting because they contain the flat and the open FLRW model, respectively, as  special cases. As such, the study of these types can be seen as a generalised study of the open and flat FLRW models. Section \ref{Sec:Anis} is a general discussion of the conditions for anisotropic hairs in the dynamical systems of the foregoing section. Throughout sections \ref{Sec:typeVIIh}-\ref{Sec:typeII} we analyse the invariant sets $\mathcal{B}$(VII$_h$), $\mathcal{B}$(VII$_0$), $\mathcal{B}$(IV) and $\mathcal{B}$(II) (in that order) by finding equilibrium sets and performing a local (and where possible: global) stability analysis.

\subsection{The general dynamical system}
\label{Subsec:GenDynSys}
Our starting point is the equations  (50)-(53) and (56)-(59) in \cite{normann18} which hold for all the Bianchi models of solvable type. Since both the variables and the derivation of the equations is explained in sufficient detail therein, we shall here suffice it to repeat only the presumably more unfamiliar part: the matter equations (eq.s \eref{FluidEqs} below) and the constraint $C_3$ (Eq.\eref{Constr3}) follow directly from the Maxwell-like equations \begin{eqnarray*}
&\mathbf{d}\mathcal{J}=0,\\
&\mathbf{d}\star\mathcal{J}=0.
\end{eqnarray*}
Here $\mathbf{d}$ is the exterior derivative and $\mathcal{J}$ is the $j$-form, taken to be constructed from an underlying $j-1$-form gauge potential.
The equations are given in a $G_2$\,-\,frame aligned with the vector $A$\setcounter{footnote}{0}\footnote{Id est: The basis-vector $\mathbf{e}_1$ is aligned with the vector $A$. The basis vectors $\{\mathbf{e}_2$,$\mathbf{e}_3\}$ now generate a 2-dimensional subgroup of the isometry group.}. They are 15 first order scalar ODEs (compactified below into 11 equations by the complex notation (boldface letters) introduced in our previous paper and further references therein). Refer to \ref{App:CompVar} for further explanation of variables.
\begin{eqnarray}\fl
\textit{j}\textrm{-form eq.s}\phantom{..}\cases{\label{FluidEqs}
V_1'=\left(q+2\Sigma_+\right)V_1-2\sqrt{3}\Re\{\mathbf{\Sigma}_1\mathbf{V}_c^*\}\,,\\
\mathbf{V}_c'=\left(q-\Sigma_+-iR_1\right)\mathbf{V}_c-\sqrt{3}\mathbf{\Sigma}_\Delta\mathbf{V}_c^*\,,\\
\Theta'=(q-2)\Theta-2AV_1\,,}\\\fl
\textrm{Einst. Eq.s}\phantom{.}\cases{\label{EinstEq}
\mathbf{\Sigma}_1'=\left(q-2-3\Sigma_+-iR_1\right)\mathbf{\Sigma}_1-\sqrt{3}\mathbf{\Sigma}_\Delta\mathbf{\Sigma}_1^*+2\sqrt{3}V_1\mathbf{V}_c\,,\\
\mathbf{\Sigma}_\Delta'=(q-2-2iR_1)\mathbf{\Sigma}_\Delta+\sqrt{3}\mathbf{\Sigma}_1^2-2\mathbf{N}_\Delta\left(iA+N_+\right)+\sqrt{3}\mathbf{V}_c^2,\\
\Sigma_+'=\left(q-2\right)\Sigma_++3\abs{\mathbf{\Sigma}_1}^2-2\abs{\mathbf{N}_\Delta}^2+\abs{\mathbf{V}_c}^2-2V_1^2\,,}\\\fl
\textrm{En. cons. }\phantom{..}
\cases{\label{EnCons}
\Omega_\Lambda'=2(q+1)\Omega_\Lambda\,,\\
\Omega_{\rm pf}'=2\left(q+1-\frac{3}{2}\gamma \right)\Omega_{\rm pf}\,,}\\\fl
\textrm{Jacobi Id.}\phantom{...}
\cases{\label{JacId2}
\mathbf{N}_\Delta'=\left(q+2\Sigma_+-2iR_1\right)\mathbf{N}_\Delta+2\mathbf{\Sigma}_\Delta N_+\,,\\
N_+'=\left(q+2\Sigma_+\right)N_++6\Re\{\mathbf{\Sigma}_\Delta^*\mathbf{N}_\Delta\}\,,\\
A'=\left(q+2\Sigma_+\right)A.}
\end{eqnarray}
In the above, $\Re$ and $\Im$ represent the real and imaginary parts, respectively; and  $'$ represents derivative with respect to the dynamical time variable $\tau$ defined through the equation
\begin{equation}
\qquad \frac{1}{H}=\frac{{\rm d} t}{\rm d\tau},
\end{equation}
where $t$ is proper time. Finally, the deceleration parameter $q$ has also been introduced in the above set of equations. It is implicitly defined by
\begin{equation}
\label{q1}
\qquad H^\prime=-(1+q)H.
\end{equation}
These dynamical equations are subject to a set of six (real) scalar constraints given by the four equations
\begin{eqnarray}
&\fl\label{Constr1} C_1=1-\Sigma_{+}^2-\abs{\mathbf{\Sigma}_{\Delta}}^2-\abs{\mathbf{\Sigma}_1}^2-\Omega_{\rm pf}-\Theta^2-V_1^2-\abs{\mathbf{V}_c}^2-\Omega_{\Lambda}-A^2-\abs{\mathbf{N}_\Delta}^2=0\,,\\
&\fl\label{Constr2} C_2=2\,\Theta V_1-2\left(A\Sigma_+-\Im\{\mathbf{\Sigma}_\Delta\mathbf{N}_\Delta^*\}\right)=0\,,\\
&\fl\label{Constr3} C_3=\sqrt{3}\mathbf{N}_\Delta^*\mathbf{V}_c-i\mathbf{V}_c^*\left(A+iN_+\right)=0\,,\\
&\fl\label{Constr4} C_4=2\,\Theta\mathbf{V}_c-\left(i\frac{N_+}{\sqrt{3}}-\sqrt{3}A\right)\mathbf{\Sigma}_1-i\mathbf{N}_\Delta\mathbf{\Sigma}_1^*=0\,.
\end{eqnarray}
In the above, $C_1$ is the so-called Hamiltonian constraint, $C_2$ and $C_4$ are the Codazzi (momentum) constraints, whereas $C_3$ comes from the matter equations, as described in the start of this sub-section. As further laid out in \cite{thorsrud19} these constraints provide a powerful check of the evolution equations, since the class of Bianchi cosmologies is a well-posed Cauchy problem. Also, the group parameter $h$ in the sets $\mathcal{B}$(VI$_h$) and $\mathcal{B}$(VII$_{h})$ is defined through 
\begin{equation}
\label{h}
A^2+h\left(3\abs{\mathbf{N}_\Delta}^2-N_+^2\right)=0. 
\end{equation}
In the above list of constraints, $C_3$ comes from the Bianchi Identity for the $j$-form and the others directly from the Einstein Field Equations. Also note that $q$ may be expressed as
\vspace{10pt}
\begin{equation}
\label{q}\fl
q=2\Sigma^2+\frac 12({3}\gamma -2)\Omega_{\rm pf}+2\Theta^2-\Omega_\Lambda,\quad\textrm{where}\quad\Sigma^2\equiv\Sigma_{+}^2+\abs{\mathbf{\Sigma}_{\Delta}}^2+\abs{\mathbf{\Sigma}_1}^2.
\end{equation}
\vspace{10pt}
\paragraph{Meaning of the variables.} The variables of the dynamical system introduced above are all expansion-normalized and have the following meaning.
\begin{itemize}
\item \textbf{Matter: } The $j$-form is given as $\mathcal{J}/(\sqrt{6}H)=(\Theta,V_1, V_2,V_3)$. Also recall that $\mathbf{V}_c=V_2+iV_3$\setcounter{footnote}{0}\footnote{Refer to \ref{App:CompVar}---or even better; to \cite{normann18}---for details and further definitions.} Furthermore, $\Omega_{\rm pf}$  is the energy density of the perfect fluid and $\Omega_\Lambda$ represents the cosmological constant.
\item \textbf{Observers: }$\mathbf{\Sigma}_1$,$\mathbf{\Sigma}_\Delta$,$\Sigma_+$ represent the shear of the congruence of observers, here chosen to be co-moving with the perfect fluid.
\item \textbf{Geometry: }$\mathbf{N}_\Delta$, $N_+$ and $A$ describe the curvature of the spatial 3-surfaces.
\item \textbf{Frame: }The quantity $R_1$ (alongside an initial angle $\phi_1$) represents the gauge freedom left in choosing the rotation (and initial orientation) around the $\mathbf{e}_1$ -axis of the orthonormal frame.
\end{itemize}
\vspace{10pt}
All the 15 variables of the dynamical system and also the frame rotation $R_1$, are expansion-normalized quantities.
\vspace{10pt}
\paragraph{State space and invariant sets.}
The equations \eref{FluidEqs}-\eref{JacId2} alongside the constraints \eref{Constr1}-\eref{Constr4} define a 9-dimensional dynamical system. The state space $D$ of this dynamical system can be divided into different invariant sets according to the aforementioned Bianchi classification. We shall adopt the notation used in \cite{thorsrud19} and in \cite{dynSys}. Bianchi set $i$ is henceforth denoted $\mathcal{B}(i)$, and the invariant sets of solvable type have the following specifications:
\newpage
\begin{itemize}
\item $\mathcal{B}$(I):\quad\quad\phantom{-}$\abs{\mathbf{N}_\Delta}^2=N_+^2/3\,=\,0\,$\phantom{0.}\phantom{0.}\quad\quad\quad,\quad\quad\quad\quad $A\,=\,0$.
\item $\mathcal{B}$(II):\,\quad\phantom{...}$\abs{\mathbf{N}_\Delta}^2-N_+^2/3\,=\,0$\quad,\quad$\abs{\mathbf{N}_\Delta}^2\,>0$\quad,\quad $A\,=\,0$.
\item $\mathcal{B}$(IV):\quad\phantom{-.}$\abs{\mathbf{N}_\Delta}^2-N_+^2/3\,=\,0$\quad,\quad$\abs{\mathbf{N}_\Delta}^2\,>0$\quad,\quad $A\,\neq\,0$.
\item $\mathcal{B}$(V):\quad\quad$\abs{\mathbf{N}_\Delta}=N_+\,=0$\quad\quad\quad\quad\quad\phantom{-}\quad,\quad\quad\quad\quad$A\,\neq\,0$.

\item $\mathcal{B}$(VI$_h$):\quad\,$\abs{\mathbf{N}_\Delta}^2-N_+^2/3\,>\,0$\quad\quad\quad\quad\phantom{-..},\quad\quad\quad\quad$A\,\neq\,0$.
\begin{itemize}
\item $\mathcal{B}$(VI$_{-1})\,=\,$B(III) (def.).
\end{itemize}
\item $\mathcal{B}$(VI$_0$):\quad\,$\abs{\mathbf{N}_\Delta}^2-N_+^2/3\,>0$\quad\quad\quad\quad\phantom{0-.},\quad\quad\quad\quad$A\,=\,0$.
\item $\mathcal{B}$(VII$_h$):\phantom{0.}$\abs{\mathbf{N}_\Delta}^2-N_+^2/3\,<\,0$\quad\quad\quad\quad\phantom{-..},\quad\quad\quad\quad$A\,\neq\,0$.
\item $\mathcal{B}$(VII$_0$):\phantom{0.}$\abs{\mathbf{N}_\Delta}^2-N_+^2/3\,<\,0$\quad\quad\quad\quad\phantom{-..},\quad\quad\quad\quad$A\,=\,0$.
\end{itemize}
That the sets $\mathcal{B}$(i) are actually invariant sets, is a well known fact. Consult for instance Chapter 15 in \cite{gron07} for a thorough treatment, or Sec. 3.2 in \cite{thorsrud19} for a recent short but sufficient treatment in our particular context. In the latter reference, Thorsrud also shows how the different Bianchi invariant sets with a $j$-form field and a perfect fluid may be further divided into disjoint subsets. In particular, for each  Bianchi type $i$ the following subsets are defined. 
\begin{eqnarray}
&\mathcal{C}^+(i)&:\quad\quad\,V_1^2>0\quad\quad\,,\quad\,\mathbf{V}_c=\mathbf{\Sigma}_1=0,\\
&\mathcal{C}^0(i)&:\quad\quad\,V_1^2=0\quad\quad\,,\quad\,\mathbf{V}_c=\mathbf{\Sigma}_1=0,\\
&\mathcal{D}^+(i)&:\quad\quad\,V_1^2+\abs{\mathbf{V}_c}^2>0,\quad\quad\textrm{+ additional constraints,}\\
&\mathcal{D}^0(i)&:\quad\quad\,V_1^2+\abs{\mathbf{V}_c}^2=0,\quad\quad\textrm{+ additional constraints.}
\end{eqnarray}
Note also the definition $\mathcal{C}$(i)$\,=\,\mathcal{C}^+$(i)$\,\cup\,\mathcal{C}^0$(i). In a similar fashion, $\mathcal{S}^+$($i$) (and $\mathcal{S}^0$($i$)) denote LRS subspaces\setcounter{footnote}{0}\footnote{In the case of VI$_0$ it is more precisely pseudo-LRS, as defined in \cite[Sec. 5.1]{thorsrud19}. The geometry of the spatial sections of these models break the LRS symmetry, but these models nevertheless allow for expansion symmetry relative to a fixed axis.} with (and without) the isotropy-breaking vector. Note that the types of subsets allowed for in each Bianchi type $\mathcal{B}$($i$) is restricted by the constraint equations, as further discussed in the reference.

\paragraph{Spatial frame.} In choosing the spatial frame it is useful to align the orthonormal frame along the eigendirections of the curvature matrix $n_{ab}$,  the vector $a_b$, or a combination thereof. In the general models, $\mathcal{B}$(VI$_0$), $\mathcal{B}$(VI$_h$), $\mathcal{B}$(VII$_0$) and $\mathcal{B}$(VII$_h$), as well as for $\mathcal{B}$(IV), this can be used to define a spatial frame unambiguously leaving no more rotational freedom. For the remaining models, there are remaining degrees of freedom: 
\begin{itemize}
\item{} $\mathcal{B}$(I) : All spatial rotations remain. 
\item{} $\mathcal{B}$(II): One spatial rotation leaving the non-zero eigenvector of $n_{ab}$ fixed. 
\item{} $\mathcal{B}$(V): One spatial rotation leaving  $a_{b}$ fixed. 
\end{itemize}

\newpage
\subsection{No-hair theorems} 
Let us also recall the no-hair theorems for these models. They determine the global behaviour of the models into the future. Both  results are proven in \cite{normann18}. 

The first is in the presence of a cosmological constant: 
\begin{thm}[First no-hair theorem]
\label{thmDeSitter}
All Bianchi invariant sets $\mathcal{B}$(I)-$\mathcal{B}$(VII$_h$) with a $j$-form, a non-phantom perfect fluid\setcounter{footnote}{0}\footnote{A perfect fluid is said to be phantom if $\gamma \,<\,0$. }  and a positive cosmological constant will be asymptotically de Sitter with $\Omega_\Lambda=1$ in the case where $\gamma\,>\,0$ (and similarly $\Omega_\Lambda+\Omega_{\rm pf}=1$ in the case where $\gamma=0$).
\end{thm}

A similar but less general theorem holds also in the case of a perfect fluid with $0\leq\,\gamma\,<\,2/3$ with a vanishing cosmological constant:
\begin{thm}[Second no-hair theorem]
\label{NoHair}
All Bianchi invariant sets $\mathcal{B}$(I)-$\mathcal{B}$(VII$_h$) with $\Omega_{\Lambda}=0$, a $j$-form, and a perfect fluid $\Omega_{\rm pf}$ with equation of state parameter $0\leq\,\gamma\,<\,2/3$ will be asymptotically quasi de Sitter with $q=\frac{3}{2}\gamma-1\,<\,0$.
\end{thm}
The two theorems above determine the behaviour into the future for these models for $\gamma\in [0,2/3)$. 

Due to the first no-hair theorem, we will henceforth {\bf assume that }$\Omega_{\Lambda}=0$, in order to determine the behavior in the case of no cosmological constant.

\paragraph{}Finally, we shall also note the function
\begin{equation}
Z=(1+\Sigma_+)^2-A^2,
\end{equation}
previously used by Hewitt and Wainwright in\,\cite{hewitt93}. It turns out that this function is useful also for our dynamical system. Starting from the general dynamical system defined by \eref{FluidEqs}-\eref{JacId2} alongside the constraints \eref{Constr1}-\eref{Constr4}, we find
\begin{equation}
Z'=-2(2-q)Z+3(1+\Sigma_+)\left((2-\gamma)\Omega_{\rm pf}+2\abs{\bm{\Sigma}_1}^2+\frac{2}{3}\abs{\bm{V}_{\rm c}}^2\right).
\end{equation}
Note that this function is monotonic for $\Omega_{\rm pf}=\bm{\Sigma}_1=\bm{V}_{\rm c}=0$.
\newpage
\section{Summary of results}
\label{Sec:Summary}
In this section we present the main results, such that the busy reader will not have to plough through all the gory details when it can be avoided. The attractors are summarized in Figures \ref{Fig:stability_BVII0}, \ref{fig:B47} and \ref{Fig:stability_BII}. 

\subsection{General results for $0\,\leq\,\gamma\,\leq\,2/3$ and also for $\Omega_{\rm pf}=0$}
Confronting the available future attractors in each Bianchi set with the No-hair theorem \ref{NoHair}, one finds that flat FLRW is the unique future attractor in the range $0\,\leq\,\gamma\,<\,2/3$ for all the invariant sets $\mathcal{B}$(II), $\mathcal{B}$(IV), $\mathcal{B}$(VII$_0$) and $\mathcal{B}$(VII$_h$).

We may also extend the proposition of Hewitt and Wainwright\cite{hewitt93} for orthogonal perfect fluid models of class B to hold also in the presence of a $j$-form fluid.

\begin{prop}[Absence of perfect fluid]\label{prop:pw} If $A\,>0$ and $\bm{\Sigma}_1=\bm{V}_{\rm c}=\Omega_{\rm pf}=0$, then any orbit in the invariant sets $\mathcal{B}$(IV), $\mathcal{B}$(V), $\mathcal{B}$(VI$_h$) and $\mathcal{B}$(VII$_h$) of the dynamical system defined by the equations \eref{FluidEqs}-\eref{Constr4} is future asymptotic to Plane Waves (PW) and past asymptotic to Jacobs' Extended Disk (JED).
\end{prop}
\begin{proof}
Use of monotonic function $Z$ in combination with tables of equilibrium sets\setcounter{footnote}{0}\footnote{The available past and future attractor candidates may be found in Table \ref{tab:StabIV} and Table \ref{tab:StabVIIhNm}. For $\mathcal{B}$(V), use Theorem 9.1 in \cite{normann18} and also Table 6 therein.} 
\end{proof}
Note that the Milne solution (M) is included in PW. From the discussion of constraints in Section \ref{Sec:constraints} this proposition generally holds for $\mathcal{B}$(IV), $\mathcal{B}$(VI$_h$) and $\mathcal{B}$(VII$_h$) in the absence of a perfect fluid. Also note that JED is the only repeller found among the equilibrium sets in all the studied models for $0<\gamma<2$, and JS the only one for $\gamma=2.$

The particular point $\gamma=2/3$ might also be of some interest. Note therefore the following proposition. 
\begin{prop}[Future asymptotes for $\gamma=2/3$]
\label{prop:genProp}
The Bianchi invariant sets $\mathcal{B}$(IV),  $\mathcal{B}$(VI$_0$), $\mathcal{B}$(VI$_h$), $\mathcal{B}$(VII$_0$) and  $\mathcal{B}$(VII$_h$)  with a perfect fluid and a $j$-form fluid are asymptotically shear-free with
\begin{equation}
1=\Omega_{\rm pf}+A^2,
\end{equation}
and $\Omega_{\rm pf}'=A'=N_+'=0$.
\end{prop}
\begin{proof}
We refer to the proof of theorem 8.1 in \cite{normann18}, of which the above proposition is an extension. The proof is similar.
\end{proof}

\paragraph{} In what follows we give particular results for each invariant set.

\subsection{The invariant set $\mathcal{B}$(VII$_0$)}
The general results cover the future asymptotes of $\gamma\,\leq\,2/3$, where FLRW is the attractor. The remaining $\gamma$-range is discussed in Section \ref{Sec:VII0} and summarized in the following theorems. 
\begin{thm}[Anisotropic hairs for $2/3<\gamma\,\leq\,2$ in $\mathcal{B}$(VII$_0$)] All $\mathcal{B}$(VII$_0$) orbits with $V_1^2\,>\,0$ and $\Omega_{\rm pf}\,>\,0$ are future asymptotic to W($\nu_1$) for $2/3<\gamma\,<\,2$ and to $T_1$\setcounter{footnote}{0}\footnote{This is the Taub section of flat space-time. Consider \cite[Sec. 6.2.2 and 9.1.6]{dynSys} for more. In this paper the Taub points are not devoted any particular attention.} ($\Sigma_+=-1$) for $\gamma=2$.
\end{thm}
We note that this implies that the $p$-form regularises the self-similarity breaking that occurs for this model when only a perfect fluid is present (see \cite{wainwright99, nilsson00} for non-tilted, and \cite{hervik06a} for tilted perfect fluid). Hence, the $\mathcal{B}$(VII$_0$) model with a $p$-form \emph{is future asymptotic to a self-similar model.} 

\begin{thm}[Past asymptotes in $\mathcal{B}$(VII$_0$)] All $\mathcal{B}$(VII$_0$) orbits with $V_1^2\,>\,0$ and $\Omega_{\rm pf}\,>\,0$ are past asymptotic to JED for $2/3\,<\,\gamma\,<\,2$ and to JS for $\gamma=2$.
\end{thm}
\begin{proof}
To prove the two above theorems, use the monotonic function $Z_6$ in \ref{App:B} and Table \ref{tab:StabVII0} of available future and past attractors. We refer to Section \ref{Sec:VII0} for details.
\end{proof}

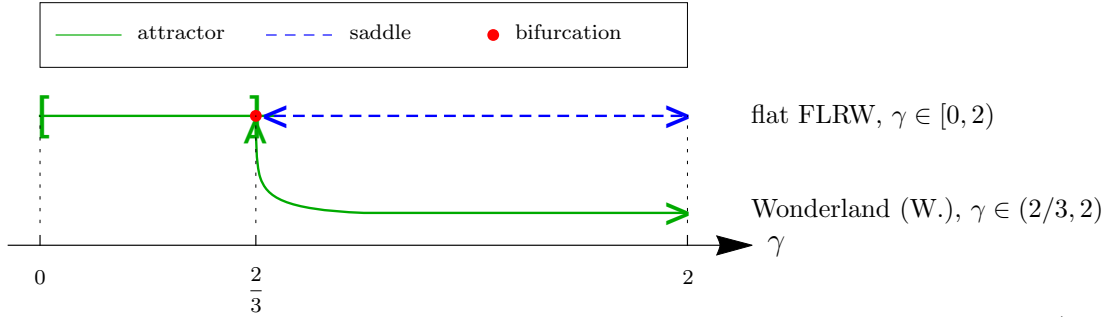
\begin{figure}[h]
	\centering	
	\begin{overpic}[width=\textwidth,tics=10]{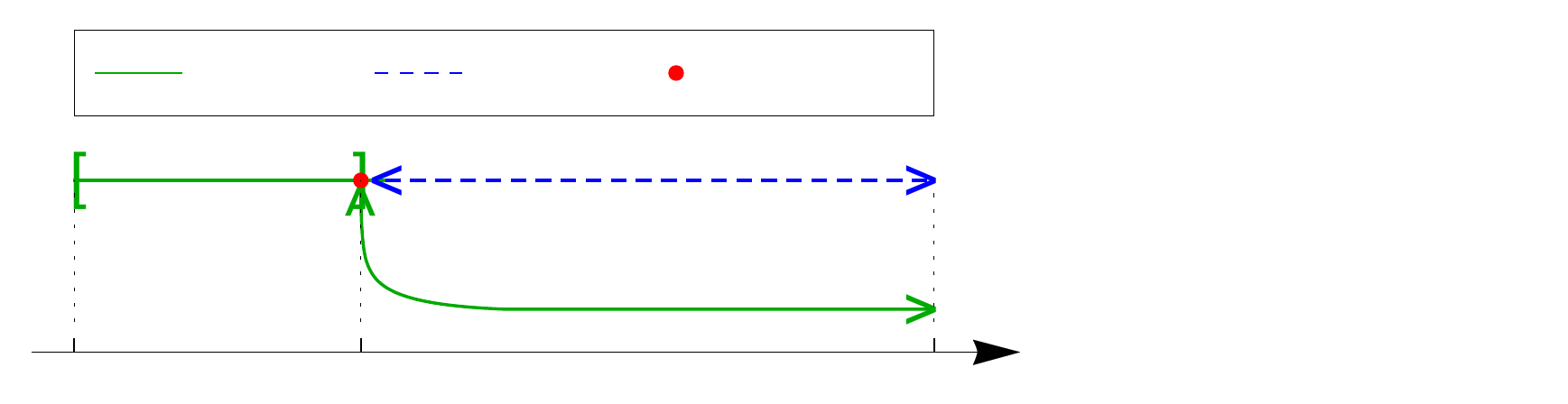}	
		\put (13,20.8) {\scriptsize attractor}
		\put (31,20.8) {\scriptsize saddle}
		\put (45,20.8) {\scriptsize bifurcation}		
		\put (65,13.7) {\footnotesize flat FLRW, $\displaystyle \gamma\in[0,2)$}
		\put (65,5.7) {\footnotesize Wonderland (W.), $\displaystyle \gamma\in\left(2/3,2\right)$}	
		\put (65.5,2.9) { $\displaystyle \gamma$ }	
		\put (3.4,0) { \scriptsize $\displaystyle 0$ }
		\put (22.3,-1) {\scriptsize $\displaystyle \frac{2}{3}$ }		
		\put (58.3,0) { \scriptsize $\displaystyle 2$ }		
	\end{overpic}	
	\caption{The diagram shows the results of the stability analysis of the set $\mathcal{B}$(VII$_0$).} 
	\label{Fig:stability_BVII0}
\end{figure}

\subsection{The invariant sets $\mathcal{B}$(IV) and $\mathcal{B}$(VII$_h$)}
The general results cover the future asymptotes for $\gamma<2/3$. Due to the lack of known monotonic functions for these sets, the results are of local character. JED ($2/3\,<\,\gamma\,<\,2$) and JS ($\gamma=2$) were the only repellers found for the sets. The attractors previously identified in $\mathcal{B}$(V) have both been found as extended families. The Plane Waves,  PW($\alpha,\nu_1,\nu^2$), exist for $\Omega_{\rm pf}=0$ and is according to Proposition \ref{prop:pw} a global attractor in the absence of the perfect fluid. In the presence of the perfect fluid, it is an attractor for $\beta_1\,>\,-\frac{3}{4}\left( \gamma-\frac{2}{3}\right)$. These results hold also for $\mathcal{B}$(IV), where $\alpha^2=1$.

The other attractor identified in $\mathcal{B}$(VII$_h$) is Wonderland, which comes as a two-paramter family $W(\kappa,\nu_1)$. It is an attractor on all of its existence, and exists for $\Omega_{\rm pf}>0.$ The subset $\nu_1=0$, called W($\kappa$), is found in $\mathcal{B}$(V), and is stable in $\mathcal{B}$(IV).

\begin{figure}[H]
\centering
\begin{tikzpicture}
\node[inner sep=0pt] (russell) at (0,0)
    {
    \begin{overpic}[width=0.7\textwidth,tics=10]{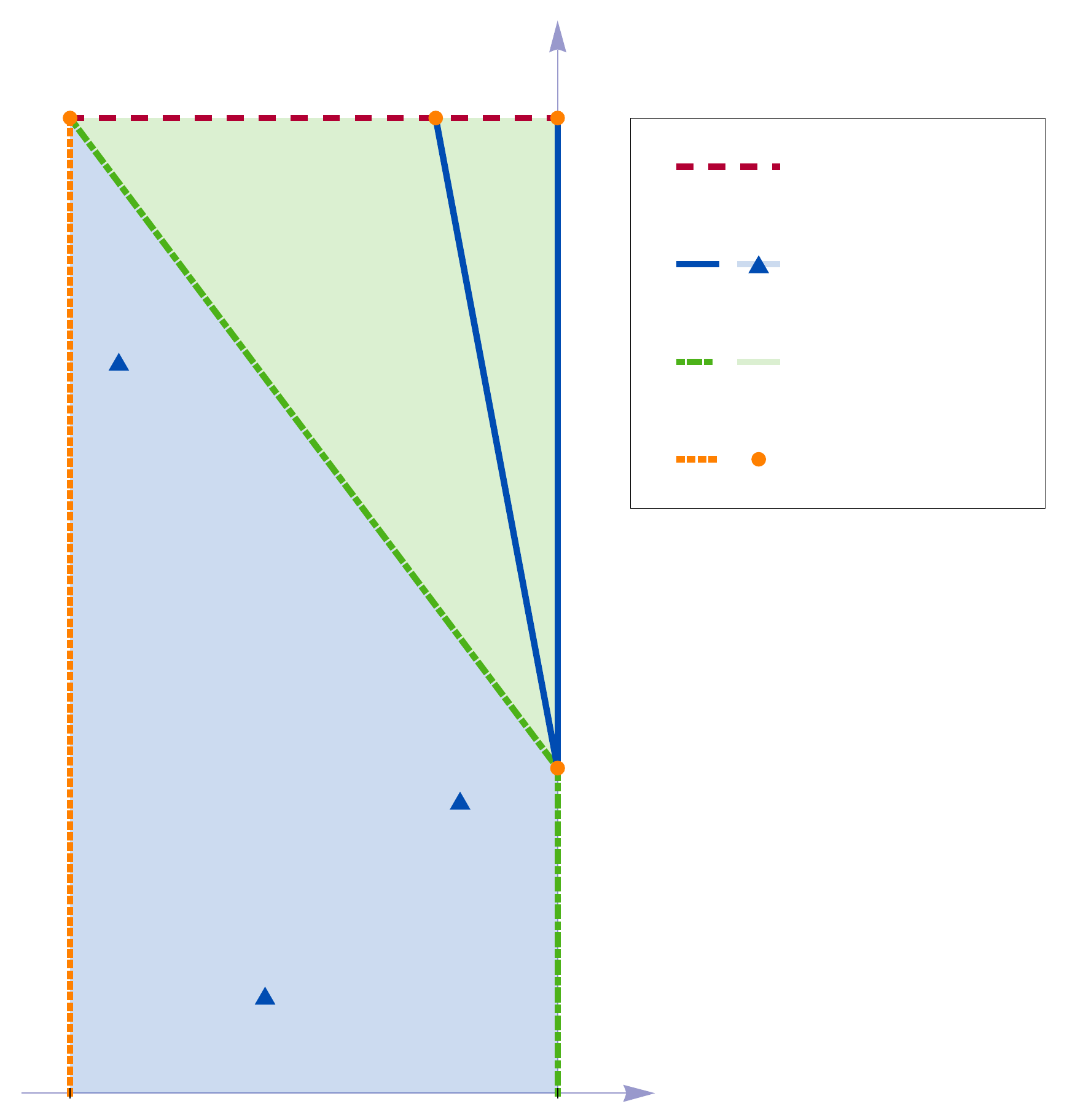}
\put (47,10) {\scriptsize{\rotatebox{90}{flat FLRW}} }
\put (47,65) {\scriptsize{\rotatebox{90}{flat FLRW}} }
\put (26,65) {\scriptsize{\rotatebox{-53.5}{Wonderland}} }
\put (17,25) {\scriptsize{\rotatebox{36.5}{Plane Waves}} }
\put (20,75) {\scriptsize{\rotatebox{36.5}{Plane Waves}} }
\put (40.6,65) {\scriptsize{\rotatebox{-78.3}{Collins-Stewart}} }
\put (25,90.5) {\scriptsize Jacobs' Sphere}
\put (2,90.5) {\scriptsize Taub ($T_1$)}
\put (7.3,50) {\scriptsize{\rotatebox{-90}{Kasner}} }
\put (50,96) { \transparent{0.5}%
{\textcolor{axisColor}{$\displaystyle \gamma$ }}}
\put (3.4,0) { \scriptsize $\displaystyle -1$ }
\put (51.7,31) {\scriptsize \transparent{0.5}%
{\textcolor{axisColor}{$\displaystyle \frac{2}{3}$ }}}
\put (51.7,88.7) {\scriptsize \transparent{0.5}%
{\textcolor{axisColor}{$\displaystyle 2$ }}}
\put (58,1.6) { \transparent{0.5}%
{\textcolor{axisColor}{\scriptsize $\displaystyle \Sigma_+$ }}}	

\put (76,84.5) {\scriptsize{Repeller}}
\put (76,76) {\scriptsize{Saddle}}
\put (64.5,76) {\scriptsize{$/$}}
\put (76,66.9) {\scriptsize{Attractor}}
\put (64.5,66.9) {\scriptsize{$/$}}
\put (76,58.5) {\scriptsize{Inconclusive}}

	\end{overpic}};
\node[] at (-0.31,-5.8) {\scriptsize $\displaystyle \Omega_{\rm pf}=1\,\uparrow$};
\node [align=left] at (3.1,-2.3) {
Schematic illustration \\
showing how the\\
different attractors \\
(green) relate to each \\
other. Note that $\gamma$ \\
increases upwards, and \\
$\Sigma_+$ is non-positive.\\
Also, Collins-Stewart\\
and FLRW are in the \\
boundaries.
};
\end{tikzpicture}
\caption{$\overline{\mathcal{B}\textrm{(IV)}}$ and $\overline{\mathcal{B}\textrm{(VII}_h)}$.}
\label{fig:B47}
\end{figure}
\subsection{The invariant set $\mathcal{B}$(II)}
The general results cover the future asymptotes of $\gamma\,\leq\,2/3$, where FLRW is the attractor. For $\gamma\,>\,2/3$, the set $\mathcal{B}$(II) is more subtle. Because of an extra degree of freedom, the vector rotations previously found in $\mathcal{B}$(I), are found also here. The following theorem has been proven to hold.
\begin{thm}[Anisotropic hairs]
\label{anisBII}
The set $\mathcal{B}$(II) with a $j$-form fluid and a perfect fluid with $\Omega_{\rm pf}>0$ is for $\Theta^2>0$ and $\gamma<2$ ($\gamma=2$) past asymptotic to JED (JS). Also, for $2/3\,<\gamma\,<\,2$ and $V_1^2>0$ it is future asymptotic to the following. 
\begin{itemize}
\item{} If $2/3<\gamma<2$ and $\Sigma_3=0$, \emph{Wonderland (W)}; 
\item{} if $2/3<\gamma\leq 6/5$ and $\Sigma_3\neq 0$,  \emph{Wonderland (W)}; 
\item{} if $6/5<\gamma< 4/3$ and $\Sigma_3\neq 0$,  \emph{the Rope (R)}; 
\item{} if $4/3\leq \gamma< 2$ and $\Sigma_3\neq 0$, \emph{the Edge (E)}.
\end{itemize}
\end{thm}
\begin{proof}
Refer to the discussion in Section \ref{Sec:typeII} for details.
\end{proof}

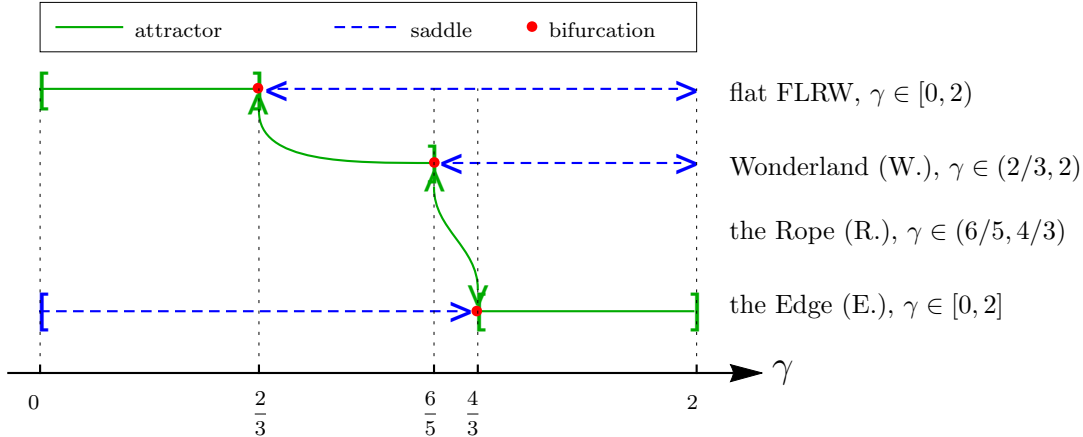
\begin{figure}[h]
	\centering	
	\begin{overpic}[width=0.9\textwidth,tics=10]{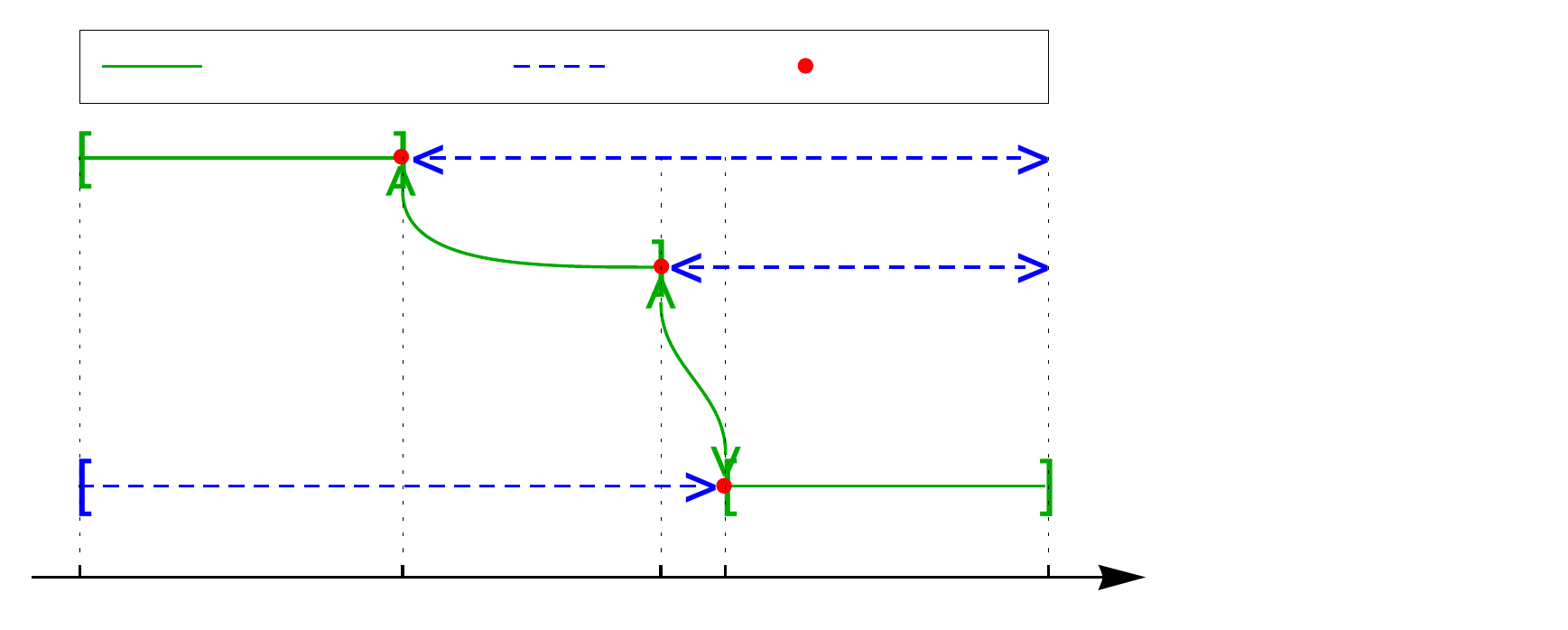}		
		\put (14,35) {\scriptsize attractor}
		\put (40,35) {\scriptsize saddle}
		\put (53,35) {\scriptsize bifurcation}		
		\put (70,29) {\footnotesize flat FLRW, $\displaystyle \gamma\in[0,2)$}
		\put (70,22) {\footnotesize Wonderland (W.), $\displaystyle \gamma\in\left(2/3,2\right)$}
		\put (70,16) {\footnotesize the Rope (R.), $\displaystyle \gamma\in\left(6/5,4/3\right)$}		
		\put (70,9) {\footnotesize the Edge (E.), $\displaystyle \gamma\in\left[0, 2\right]$}	
		\put (74,3) {\large $\displaystyle \gamma$ }	
		\put (3,0) { \scriptsize $\displaystyle 0$ }
		\put (25,-1) {\scriptsize $\displaystyle \frac{2}{3}$ }	
        \put (41,-1) {\scriptsize $\displaystyle \frac{6}{5}$ }	
		\put (45,-1) {\scriptsize $\displaystyle \frac{4}{3}$ }
		\put (65,0) { \scriptsize $\displaystyle 2$ }		
	\end{overpic}	
	\caption{A schematic stability diagram for $\mathcal{B}$(II). Closed and open intervals are indicated using $[$ , $]$ and $($ , $)$, respectively.} 
	\label{Fig:stability_BII}
\end{figure}
\section{A general discussion of constraints}
\label{Sec:constraints}
Throughout the analysis, we will use the Hamiltonian constraint  \eref{Constr1} to remove one variable from the dynamical system\setcounter{footnote}{0}\footnote{In our case, we consequently remove $\Omega_{\rm pf}$}. The constraint \eref{Constr2} will be treated on an individual basis for each equilibrium set. The remaining two constraints, however, equations \eref{Constr3} and \eref{Constr4}, allow for a collective analysis and implementation, as discussed below.
\label{Sec:constraints}
\subsection{Constraint $C_3$}
The constraint \eref{Constr3} may be rewritten to 
\begin{equation}
\label{Constr1R}
\left[ {\begin{array}{cc}
\sqrt{3}\mathbf{N}_\Delta^*     &N_+-iA\\
N_++iA& \sqrt{3}\mathbf{N}_\Delta   \\
\end{array} } \right]
\left[ {\begin{array}{cc}
\mathbf{V}_c\\
\mathbf{V}_c^*\\
\end{array} } \right]
=
\left[ {\begin{array}{cc}
0\\
0\\
\end{array} } \right],
\end{equation}
where both the constraint and its complex conjugate are written down. The system is on the form $Mx=0$ where $M\,\in\,\mathcal{M}_{2\times2}$. There are two possibilties, either $\det(M)\neq 0~ \Rightarrow ~\mathbf{V}_c=0$, or  $\det(M)=0$. The latter permits a non-zero ${\bf V}_c$ as long as ${\bf V}_c$ is a zero-eigenvector.  Explicitly, we have
\begin{equation}
\label{CharEqConstr1}
\det(M)= 3\abs{\mathbf{N}_\Delta}^2-N_+^2-A^2.
\end{equation}
 Hence, the two possibilities are:
\begin{eqnarray*}
&\fl\textbf{Option 1: }\quad\mathbf{V}_c=0\quad\quad\quad&\textrm{(allowed in all sets)}\\
&\fl\textbf{Option 2: }\quad \mathbf{V}_c\neq 0, ~~3\abs{\mathbf{N}_\Delta}^2-N_+^2-A^2=0\quad\quad\quad&\textrm{(allowed in $\mathcal{B}$(I), $\mathcal{B}$(II) and $\mathcal{B}$(III).} 
\end{eqnarray*}
Note that the lone non-zero eigenvalue in $\mathcal{B}$(II) and $\mathcal{B}$(III) is
\begin{equation}
\lambda_{\rm II/III}=2\sqrt{3}N_-,
\end{equation} 
which vanishes whenever $N_-=0$. We shall not need any further discussion of this constraint, as it turns out that also in $\mathcal{B}$(II) (which is the only one studied in this paper among the option 2 types) will we be able to set $\mathbf{V}_c=0$ using a rotation of frame. It is instructive at this point to go to the real plane and look at the eigenvectors of the system with only one non-zero eigenvalue. Substituting one of the equations into the other one ends up with 
\begin{eqnarray}
&\left(3(N_-^2+N_\times^2)-N_+^2-A^2\right)V_3=0\\
&(\sqrt{3}N_\times+A)V_2=(\sqrt{3}N_--N_+)V_3.
\end{eqnarray}
From the above we can clearly see (first equation) that in $\mathcal{B}$(I), $\mathcal{B}$(II) or $\mathcal{B}$(III), $V_3$ is a free parameter (bound only by the Hamiltonian constraint). And so (by the second equation) must $V_2$ be also.

\subsection{Constraint $C_4$}
Note the simplification of \eref{Constr4} that follows directly from $\mathbf{V}_c=0$:
\begin{equation}
\left[ {\begin{array}{cc}
\sqrt{3}\mathbf{N}_\Delta^*&N_+-3iA\\
N_++3iA& \sqrt{3}\mathbf{N}_\Delta\\
\end{array} } \right]
\left[ {\begin{array}{cc}
\mathbf{\Sigma}_1\\
\mathbf{\Sigma}_1^*\\
\end{array} } \right]
=
\left[ {\begin{array}{cc}
0\\
0\\
\end{array} } \right]
\end{equation}
Again the system is on the form $Nx=0$ where $N\,\in\,\mathcal{M}_{2\times2}$. Again the solutions are either $\det(N)\neq 0~ \Rightarrow ~\mathbf{\Sigma}_1=0$, or $\det(N)=0$. The latter permits for a non-zero $\mathbf{\Sigma}_1$. Explicitly, we have
\begin{equation}
\label{CharEqConstr2}
\det(N)= 3\abs{\mathbf{N}_\Delta}^2-N_+^2-9A^2.
\end{equation}

The \textbf{Option 1} from above can thus be further divided into  two subclasses. We write (including the allowed sets)
\begin{eqnarray*}
&\fl\textbf{Option 1a: }\quad\mathbf{V}_c=0\quad\textrm{and}\quad\mathbf{\Sigma}_1=0\quad\phantom{..}\quad\quad\quad&\textrm{(allowed in all sets)}.\\
&\fl\textbf{Option 1b: }\quad\mathbf{V}_c=0\quad\textrm{and}\quad\mathbf{\Sigma}_1\neq 0\quad\quad\quad\quad&\textrm{(allowed in $\mathcal{B}$(I), $\mathcal{B}$(II) and $\mathcal{B}$(VI$_{-1/9}$).} 
\end{eqnarray*}

\subsection{The sets $\mathcal{B}$(I), $\mathcal{B}$(II), $\mathcal{B}$(III) and $\mathcal{B}$(VI$_{-1/9}$)}
\paragraph{}From the Options 1b and 2 above we find that the sets $\mathcal{B}$(I), $\mathcal{B}$(II), $\mathcal{B}$(III) and $\mathcal{B}$(VI$_{-1/9}$) require special and separate treatments. Also invoking the constraint $C_2$ we reach the following conclusions.
\begin{itemize}
\item$\mathcal{B}$(I): $\mathbf{V}_c=V_1=0\quad\textrm{or}\quad\Theta=0$.
\item$\mathcal{B}$(II): Here, $A=0$ so we still have some gauge freedom left in rotating the spatial frame. The non-zero eigenvector of the matrix $M$ (refer to discussion above) defines a spatial direction, and the $(V_1,{\bf V}_c)$ needs to be orthogonal to this. By using a spatial rotation orthogonal to this eigendirection of $M$ we can align $(V_1,\mathbf{V}_c)$, for instance, along  $\mathbf{e}_1$. Hence, in this gauge we have $\mathbf{V}_c=0$. 
\item$\mathcal{B}$(III): This special case ($h=-1$) of $\mathcal{B}$(VI$_h$) needs a different treatment because it allows for an extra ${\bf V}_c$ degree of freedom. A general study of shear-free solutions with $p$-form gauge fields has revealed that this is the only type in which a shear-free solution with a lower-bounded Hamiltonian exists \cite{thorsrud18}. Note also the recent work~\cite{thorsrud20}.
\item $\mathcal{B}$(VI$_{-1/9}$): The particular value $h=-1/9$ allows for an extra shear degree of freedom, and again particular treatment is needed. It is called the exceptional case, denoted $\mathcal{B}$(VI$^*_{-1/9}$) whenever this extra degree of freedom is included. For the perfect fluid case, consult Chapter 8 in \cite{dynSys} for the non-tilted case and \cite{hervik07} for the tilted case.
\end{itemize}

\section{Dynamical system for $\mathcal{B}$(IV), $\mathcal{B}$(VI$_0$), $\mathcal{B}$(VI$_h$), $\mathcal{B}$(VII$_0$) and $\mathcal{B}$(VII$_h$)}
\label{Sec:SysEqIVEtc}
Based on the discussion of constraints in the previous section, we may reduce the dynamical system introduced in Section \ref{Intro} to a general set of equations that cover the Bianchi type A invariant sets $\mathcal{B}$(VI$_0$) and $\mathcal{B}$(VII$_0$) alongside the Bianchi type B invariant sets $\mathcal{B}$(IV), $\mathcal{B}$(VI$_h$) and $\mathcal{B}$(VII$_h$). These are the sets for which the curvature can be used to define the spatial frame unambiguously (once the gauge is chosen). This is done in the following. Note that the relations between the complex variables used in previous sections and the real variables used here (and in the rest of the paper) is given in \ref{App:CompVar}.

\subsection{System before gauge choice}
\label{Sec:BefGChice}
The five above mentioned sets are Option 1a types (refer to the discussion of constraints in the previous section) and hence we must have $\mathbf{V}_c=\mathbf{\Sigma}_1=0$ to fulfill the constraints \eref{Constr3} and \eref{Constr4}. The general system of equations before gauge choice therefore reads as follows.
\begin{eqnarray}
&\fl\textit{j}\textrm{-form eq.s }\quad\quad\cases{\label{FluidEqsOpt1a}
V_1'=\left(q+2\Sigma_+\right)V_1,\\
\Theta'=(q-2)\Theta-2AV_1\,,}&\\
&\fl\textrm{Einst. eq.s }\quad\phantom{00}\cases{\label{EinstEqOpt1a}
\Sigma_-'=(q-2)\Sigma_-+2R_1\Sigma_\times+2 (A N_\times -N_- N_+)\\
\Sigma_\times'=(q-2)\Sigma_\times-2R_1\Sigma_--2 (A N_- +N_\times N_+)\\
\Sigma_+'=\left(q-2\right)\Sigma_+-2\left(N_-^2+N_\times^2+V_1^2\right)\,,}&\\
&\fl\textrm{En. cons. }\quad\quad\phantom{0}
\cases{\label{EnConsOpt1a}
\Omega_{\rm pf}'=2\left(q+1-\frac{3}{2}\gamma \right)\Omega_{\rm pf}\,,}&\\
&\fl\textrm{Jacobi id. }\quad\quad\phantom{.}
\cases{\label{JacId2Opt1a}
N_-'=\left(q+2\Sigma_+\right)N_-+2(R_1 N_\times+\Sigma_- N_+)\,,\\
N_\times'=\left(q+2\Sigma_+\right)N_\times-2(R_1 N_--\Sigma_\times N_+)\,,\\
N_+'=\left(q+2\Sigma_+\right)N_++6\left(\Sigma_-N_-+\Sigma_\times N_\times\right)\,,\\
A'=\left(q+2\Sigma_+\right)A.}&
\end{eqnarray}
The remaining constraints \eref{Constr1} and \eref{Constr2} now read, respectively,
\begin{eqnarray}
&\fl\label{Constr1Option1a} C_1=1-\Omega_{\rm pf}-\Sigma_{+}^2-\Sigma_-^2-\Sigma_\times^2-\Theta^2-V_1^2-A^2-N_-^2-N_\times^2=0\,,\\
&\fl\label{Constr2Option1a} C_2=\Theta V_1-A\Sigma_+-N_\times\Sigma_-+N_-\Sigma_\times=0\,.
\end{eqnarray}
\paragraph{Useful observation:} Also note that following directly from the system of equations above is the result that $V_1/A$ is a constant of motion;
\begin{equation}
\label{useful}
\left(\frac{V_1}{A}\right)'=0,
\end{equation}

\subsection{Choosing gauge}
\label{Subsec:Gauge}
The question now becomes that of choosing gauge \cite{coley05}. We will make use of two different choices in our analysis.
\begin{itemize}
\item Use the gauge freedom to diagonalize  $N_{ab}$. This means we let $N_+=\sqrt{3}\alpha\Re\{\mathbf{N}_\Delta\}$, by appropriately choosing $R_1$. We find 
\begin{equation*}
R_1=\sqrt{3}\alpha \Sigma_\times\phantom{00000}\textrm{and}\phantom{00000}N_+=\sqrt{3}\alpha N_-\quad\quad\quad\quad\textbf{($N_-$\,-\,gauge).}
\end{equation*}
for some function $\alpha(\tau)$. If we use our remaining freedom (choosing $\phi_1(\tau=0)$) to say that $N_\times(\tau=0)=0$, then $N_\times$ will remain zero. Such a choice is possible\setcounter{footnote}{0}\footnote{Specifically it may be shown that the choice will obey
\begin{equation*}
\phi_1(\tau=0)=-\frac{1}{2}\tan\left(\frac{\tilde{\Sigma}_\times}{\tilde{\Sigma}_-}\right)
\end{equation*}
where $\tilde{\Sigma}_\times,\,\tilde{\Sigma}_-$ are variables referring to the frame following gyroscopes.} 
\item A second choice that proves useful whenever $N_\times=N_+=0$ is
\begin{equation*}
R_1=0.\quad\quad\quad\quad\quad\quad\quad\quad\quad\quad\quad\quad\quad\quad\quad\quad\quad\quad\quad\quad\textbf{($F$\,-\,gauge).}
\end{equation*}
In this case we should keep in mind that we have a constant gauge freedom left (namely $\phi_1(\tau=0)$).
\end{itemize}  
\subsection{Equilibrium sets and scalars}
The definition of an equilibrium point is
\begin{dfn}[Equilibrium point] 
An equilibrium point P is a set on which all scalars are constants on $P$ as functions of $\tau$. 
\end{dfn}
It is in place, therefore, at this point with a reminder of what the scalars of the dynamical system are. In particular, for the $\mathcal{B}$(VI$_h$) and $\mathcal{B}$(VII$_h$) systems we shall note that $N_-$ and $N_\times$, $\Sigma_-$ and $\Sigma_\times$ are not scalars \cite{normann18}. Rather, from the complex entities $\bm{\Sigma}_\Delta$ and $\bm{N}_\Delta$ left in the dynamical system, we may construct three independent scalars as follows.
\begin{eqnarray}
&\label{sig}\sigma^2\,\equiv\,\bm{\Sigma}_\Delta\bm{\Sigma}_\Delta^*=\Sigma_-^2+\Sigma_\times^2,\\
&\label{nu}\nu^2\,\equiv\,\bm{N}_\Delta\bm{N}_\Delta^*=N_-^2+N_\times^2,\\
&\label{delta}\delta^2\,\equiv\,\bm{N}_\Delta\bm{\Sigma}_\Delta^*=N_-\Sigma_-+N_\times\Sigma_\times+i(N_\times\Sigma_--N_-\Sigma_\times).
\end{eqnarray}
Equilibrium sets in these spaces may therefore have evolving $\Sigma_-,\Sigma_\times,N_-, N_\times$, as long as $\sigma,\nu,\delta$ are constants on the motion.

\paragraph{}In the dynamical systems analysis, we will use both the $N_-$\,-\,gauge and the $F$\,-\,gauge, so in the following subsections we will spell out the dynamical systems for these two choices.
\subsection{Dynamical system in $N_-$\,-\,gauge}
\label{Sec:DynSysN}
Doing the math, one finds that the group parameter $h$ \eref{h} in $\mathcal{B}$(VI$_h$) and $\mathcal{B}$(VII$_h$) is now given by 
\begin{equation}
A^2=3h\left(\alpha^2-1\right)N_-^2\quad\rightarrow\quad h=\frac{1}{\left(\alpha^2-1\right)}\left(\frac{A}{\sqrt{3}N_-}\right)^2
\end{equation}
From the type specifications (of the invariant Bianchi sets) detailed out in a previous section we hence find the following.

\begin{eqnarray}
&\fl\mathcal{B}(\textrm{VII}_h)\,:\quad\quad &\alpha^2>1\,\quad\rightarrow\quad h>0,\\
&\fl\mathcal{B}(\textrm{VI}_h)\phantom{0}:\quad\quad &\alpha^2<1\,\quad\rightarrow\quad h<0,\\
&\fl\mathcal{B}(\textrm{IV})\phantom{00}:\quad\quad &\alpha^2=\,1\quad\rightarrow\quad h =\infty.
\end{eqnarray}
The dynamical system described in Section \ref{Intro} is now 7 -dimensional and takes the following form.
\begin{eqnarray}
&\fl\textit{j}\textrm{-form eq.s }\quad\quad\cases{\label{FluidEqsOpt1a2}
V_1'=\left(q+2\Sigma_+\right)V_1,\\
\Theta'=(q-2)\Theta-2A V_1\,,
}&\\
&\fl\textrm{Einst. eq.s }\quad\phantom{00}\cases{\label{EinstEqOpt1a2}
\Sigma_-'=(q-2)\Sigma_-+2\sqrt{3}\alpha(\Sigma_\times^2-N_-^2),\\
\Sigma_\times'=(q-2-2\sqrt{3}\alpha\Sigma_-)\Sigma_\times-2 A N_-,\\
\Sigma_+'=\left(q-2\right)\Sigma_+-2\left(N_-^2+V_1^2\right)\,,
}&\\
&\fl\textrm{En. cons. }\quad\quad\phantom{0}
\cases{\label{EnConsOpt1a2}
\Omega_{\rm pf}'=2\left(q+1-\frac{3}{2}\gamma \right)\Omega_{\rm pf}\,,}&\\
&\fl\textrm{Jacobi id. }\quad\quad\phantom{.}
\cases{\label{JacId2Opt1a2}
N_-'=\left(q+2\Sigma_++2\sqrt{3}\alpha\Sigma_-\right)N_-\,,\\
\alpha'=2\sqrt{3}\Sigma_-(1-\alpha^2)\,,\\
A'=(q+2\Sigma_+)A.}&
\end{eqnarray}
Note that the equation for $\alpha^\prime$ is found by employing the defining condition $N_+=\sqrt{3}\alpha N_-$ for the $N_-$\,-\,gauge. Additionally, the two remaining constraints are
\begin{eqnarray}
&\fl\label{Constr1Option1aN} C_1=1-\Omega_{\rm pf}-\Sigma_{+}^2-\Sigma_-^2-\Sigma_\times^2-\Theta^2-V_1^2-A^2-N_-^2-N_\times^2=0\,,\\
&\fl\label{Constr2Option1aN} C_2=\Theta V_1-A\Sigma_++N_-\Sigma_\times=0\,.
\end{eqnarray}
\paragraph{Symmetries.} The dynamical system described in this section has the following symmetries:
\begin{eqnarray}
&\label{NSym1}\fl(V_1',V_1,\Theta',\Theta)\,\rightarrow\,(-V_1',-V_1,-\Theta',-\Theta),\\
&\label{NSym2}\fl (\Sigma_-',\alpha',\Sigma_-,\alpha)\,\rightarrow\,(-\Sigma_-',-\alpha',-\Sigma_-,-\alpha),\\
&\label{NSym3}\fl(N_-',N_-,\Sigma_\times',\Sigma_\times)\rightarrow(-N_-',-N_-,-\Sigma_\times',-\Sigma_\times).
\end{eqnarray}
\paragraph{Remark.} For any equilibrium point in $\mathcal{B}$(VI$_h$) and $\mathcal{B}$(VII$_h$) , it is clear from the $\alpha'$-equation that we must have $\Sigma_-=0$. For $\mathcal{B}$(IV), however, this is not necessarily true, since $\alpha^2=1$.

\subsection{Fermi\,-\,gauge ($F$\,-\,gauge)}
\label{Sec:DynSysF}
Another choice is the so-called $F$\,-\,gauge; the frame where $R_1=0$ and hence one (or more) of the axis in the tetrad follows the gyroscopes. The dynamical system of Section \ref{Intro} is now 8-dimensional, and takes the following form.
\begin{eqnarray}
&\fl\textit{j}\textrm{-form eq.s }\quad\quad\cases{\label{FluidEqsOpt1aR}
V_1'=\left(q+2\Sigma_+\right)V_1,\\
\Theta'=(q-2)\Theta-2AV_1\,,
}&\\
&\fl\textrm{Einst. eq.s }\quad\phantom{00}\cases{\label{EinstEqOpt1aR}
\Sigma_-'=(q-2)\Sigma_-+2 (A N_\times -N_- N_+)\\
\Sigma_\times'=(q-2)\Sigma_\times-2 (A N_- +N_\times N_+)\\
\Sigma_+'=\left(q-2\right)\Sigma_+-2\left(N_-^2+N_\times^2+V_1^2\right)\,,
}&\\
&\fl\textrm{En. cons. }\quad\quad\phantom{0}
\cases{\label{EnConsOpt1aR}
\Omega_{\rm pf}'=2\left(q+1-\frac{3}{2}\gamma \right)\Omega_{\rm pf}\,,
}&\\
&\fl\textrm{Jacobi id. }\quad\quad\phantom{.}
\cases{\label{JacId2Opt1aR}
N_-'=\left(q+2\Sigma_+\right)N_-+2\Sigma_- N_+\,,\\
N_\times'=\left(q+2\Sigma_+\right)N_\times+2\Sigma_\times N_+\,,\\
N_+'=\left(q+2\Sigma_+\right)N_++6\left(\Sigma_-N_-+\Sigma_\times N_\times\right)\,,\\
A'=\left(q+2\Sigma_+\right)A.}&
\end{eqnarray}
The two remaining constraints are 
\begin{eqnarray}
&\fl\label{Constr1Option1aF} C_1=1-\Omega_{\rm pf}-\Sigma_{+}^2-\Sigma_-^2-\Sigma_\times^2-\Theta^2-V_1^2-A^2-N_-^2-N_\times^2=0\,,\\
&\fl\label{Constr2Option1aF} C_2=\Theta V_1-A\Sigma_+-N_\times\Sigma_-+N_-\Sigma_\times=0\,.
\end{eqnarray}
As already mentioned there is a constant gauge freedom left (choosing of initial angle $\phi_1(\tau=0)$) that we need to be aware of.

\paragraph{Symmetries} The dynamical system described in this section has the following symmetry:
\begin{eqnarray}
\label{FSym1}
(V_1',V_1,\Theta',\Theta)\,\rightarrow\,(-V_1',-V_1,-\Theta',-\Theta)
\end{eqnarray}

\subsection{Dynamical systems analysis}
Based on the first order stability analysis we adopt the ordinary procedure of dividing hyperbolic equilibrium points (sets) into three classes. We use the following language.
\begin{itemize}
\item Attractor: All eigenvalues have negative real parts.
\item Saddle: A mixture of eigenvalues with positive and negative real parts.
\item Repeller: All eigenvalues have positive real parts.
\end{itemize} 
In the case of sets, one must bear in mind that the interesting perturbations are those that are orthogonal to the set. 

For the non-hyperbolic equilibrium points (sets), we sometimes adopt \textit{center-manifold analysis}. The reader is referred to the classic text by Perko \cite{perko} for a thorough introduction to dynamical systems in general, and to chapter 4 of \cite{dynSys} for a brief introduction to dynamical systems applied to cosmology.

\section{Anisotropic hairs}
\label{Sec:Anis}
In order to have expanding, self-similar space-times with pertaining anisotropies, one must have a source. In the general dynamical system of Section \ref{Sec:BefGChice}, where $\bm{\Sigma}_1=\bm{V}_{\rm c}=0$, two vectors are available: (i) The geometric option is $A$, the expansion-normalized trace of the structure coefficients of the Lie algebra, and (ii) the remaining vector part of the matter sector: $V_1$. In both cases, we must have 
\begin{equation}
q=-2\Sigma_+\,,
\end{equation}
in order for the derivatives of either source to vanish non-trivially. In a self-similar space-time, all scalars will be constant. Henceforth, starting from the equations of the dynamical system, we can make further restrictions. First, $N_+'=0$ implies that $\Re\{\delta^2\}$, as defined in \eref{delta}, must vanish. A parametric choice that enforces $\Re\{\delta^2\}=0$ is 
\begin{eqnarray}
\label{nn}
 N_-=\nu_2\quad,\quad N_\times=\nu_3\\
\label{ss}\Sigma_-=-\kappa\nu_3\quad,\quad\Sigma_\times=\kappa\nu_2.
\end{eqnarray}
From the definitions \eref{sig}-\eref{delta} we now find 
\begin{eqnarray}
\fl\nu^2=\nu_2^2+\nu_3^2\quad\quad,\quad\quad\sigma^2\,=\kappa^2\nu^2\quad\quad,\quad\quad &\delta^2\,=-i\kappa^2\nu^2.
\end{eqnarray}
Hence, to be at an equilibrium point, we must require $\nu'=\kappa'=0$. Inserting the parameterization \eref{nn}-\eref{ss} into the dynamical system of Section \ref{Sec:BefGChice} (also changing names such that $N_+=\nu_1$ and $\Sigma_+=\beta_1$) the parameter derivatives are found to fulfill the equations 
\begin{eqnarray}
\nu_2^\prime=-2\kappa\nu_1\nu_3,\\
\nu_3^\prime=2\kappa\nu_1\nu_2,\\
\label{nnl1}\nu_2\nu_3\cdot\kappa'=2 \nu _1 \nu _2^2\left(1-\kappa^2\right)-2\left(A+\kappa(\beta _1+1) \right)\nu_2\nu_3,\\
\label{nnl2}\nu_2\nu_3\cdot\kappa'=-2 \nu _1 \nu _3^2\left(1-\kappa^2\right)-2\left(A+\kappa(\beta _1+1) \right)\nu_2\nu_3.
\end{eqnarray}
Equating the two latter equations above (following from the $\Sigma_-^\prime$- and $\Sigma_\times$-equation, respectively) we find the algebraic constraint
\begin{equation}
\label{cond}
\nu_1\left(1-\kappa^2\right)\nu^2=0.
\end{equation}
Consequently we are left with the following three options.
\begin{enumerate}
\item $\kappa^2=1$. Going back to equations \eref{nnl1}, and requiring a positive $A$, we must in this case have $\kappa=-1$ and hence $A=\beta_1+1$.
\item $\nu=0$. In this case $\kappa^\prime$ remains unspecified from eqs. \eref{nnl1}-\eref{nnl2} above. This, however, must be seen as an artefact of the parameterization, and represents no real physical degree of freedom. 
\item $\nu_1=0.$ In this case we find $\kappa^\prime=-2\left(A+\kappa(\beta _1+1) \right)$. Since $\kappa^\prime=0$ is required for $(\sigma^{2})^\prime=0$, we must have
\begin{equation}
A=-\kappa(\beta _1+1).
\end{equation}
\end{enumerate}
Furthermore, the constraint $C_2$, the equations $\Omega_{\rm pf}^\prime=0$ and $\beta_1^\prime=0$ and $\Theta^\prime=0$ and the definition of $q$ (eq. \eref{q}), give the following restrictions, respectively.
\begin{eqnarray}
\label{cond2}A \beta _1-\Theta  V_1=\kappa\nu^2,\\
\label{cond1}0=2\left(q+1-\frac{3}{2}\gamma \right)\Omega_{\rm pf}\\
\label{cond3}0=-2 \left(\beta _1^2+\beta _1+\nu^2+V_1^2\right),\\
\label{cond4}0=-2 \left(A V_1+\left(\beta _1+1\right)\Theta\right),\\
\label{cond5}-2\Sigma_+(1+\Sigma_{+})=2\left(\Theta^2+\kappa^2\nu^2\right)+\frac{1}{2}(3\gamma -2)\Omega_{\rm pf}.
\end{eqnarray}
From these one may derive two general sets of equilibrium points. The so-called Plane Waves and Wonderland. In the following two subsections we describe these sets more carefully. Throughout we use the same parameterization as in this section, unless otherwise is explicitly specified. In particular,
\begin{equation}
\label{param}
(N_+,N_-,N_\times,\Sigma_+,\Sigma_-,\Sigma_\times)\rightarrow(\nu_1,\nu_2,\nu_3,\beta_1,-\kappa\nu_3,\kappa\nu_2)
\end{equation}

\subsection*{Plane Waves, PW$(\beta_1,\nu_1,\nu^2)$}
With option (i) above, $\kappa$ is fixed. Specifically, $\kappa=-1$ \setcounter{footnote}{0}\footnote{One could have $\kappa=1$ as well, but we have (without loss of generality) chosen to align our frame along $A$, so $\kappa$ must be negative.}. The physical freedom is now in the tuple $\nu_1,\nu$. It turns out that the remaining equations \eref{cond1}-\eref{cond5} solve to give the so-called Plane Waves equilibrium set. With the parameterization as before, the further specifications of PWs is as follows.
\begin{eqnarray}
\label{PW}\fl
\kappa=-1\quad,\quad A=1+\beta_1\quad,\quad V_1^2=-\beta_1(1+\beta_1)-\nu^2\quad,\quad\Theta=-V_1.
\end{eqnarray}
It is straight forward to verify that for $\nu_1\,\neq\,0$, the decoupled $\nu_2^\prime$ and $\nu_3^\prime$ equations are solved by
\begin{equation}
\nu_2=\nu\sin(2\nu_1\tau)\quad\textrm{and}\quad\nu_3=\nu\cos(2\nu_1\tau)
\end{equation}
where $\nu$ is the constant of motion defined in eq. \eref{nu}. The family $PW(\beta_1,\nu_1,\nu^2)$ stretches over several Bianchi invariant sets, and may be divided into different invariant subsets using the Bianchi classification. We shall only name the following two subsets.
\begin{itemize}
\item $\mathcal{S}^+$(VII$_h$)$\,\supset\,\mathcal{P}_{PW(\beta_1,\nu_1)}\equiv \lim_{\nu\rightarrow 0}PW(\beta_1,\nu_1,\nu^2)$.
\item $\mathcal{S}^+$(V)$\,\supset\,\mathcal{P}_{M}\equiv \lim_{\beta_1,\nu_1,\nu\rightarrow 0}PW(\beta_1,\nu_1,\nu^2)$.
\end{itemize}
$M$ is here the Milne exact vacuum solution, and we have $M\,\subset\,PW(\beta_1,\nu_1)\,\subset\,PW(\beta_1,\nu_1,\nu^2)$.
\paragraph{}As a final remark; the options (ii) and (iii) will only produce parts of these invariant subspaces of PW.

\subsection*{The Wonderland fabric, $W(\kappa,\nu_1,\nu^2)$}
Starting from Options (ii) $\nu_1=0$ and (iii) $\nu=0$ gives another family of equilibrium sets: Wonderland, denoted $W(\lambda,\nu_1,\nu^2)$.\setcounter{footnote}{0}\footnote{This new family of equilibrium points is an extension to the previously found $\mathcal{B}$(V)-equilibrium set with the same name \cite{normann18}}. With the parameterization as before, it has the following specifications:
\begin{eqnarray}\fl
\beta_1=\frac{1}{4}(2-3\gamma)\quad,\quad \nu_1\nu^2=0\\
\fl A=-\kappa(1+\beta_1)\quad,\quad V_1^2=-\beta_1(1+\beta_1)-\nu^2\quad,\quad\Theta=\kappa V_1.
\end{eqnarray}
Note that instead of having $\kappa$ fixed, as with the Plane waves, we now have a one-to-one relation between the shear $\beta_1$ and $\gamma$. The family $W(\kappa,\nu_1,\nu^2)$ may be divided into several subsets that belong to different invariant sets. They are as follows.
\begin{itemize}
\item $\mathcal{S}^+(I)\,\supset\,\mathcal{P}_{W}\equiv \lim_{\kappa,\nu_1,\nu\rightarrow 0}W(\kappa,\nu_1,\nu^2)$.
\item $\mathcal{S}^+(V)\,\supset\,\mathcal{P}_{W(\kappa)}\equiv \lim_{\nu_1,\nu\rightarrow 0}W(\kappa,\nu_1,\nu^2)$.
\item $\mathcal{S}^+(\textrm{VII}_h)\,\supset\,\mathcal{P}_{W(\kappa,\nu_1)}\equiv \lim_{\nu\rightarrow 0}W(\kappa,\nu_1,\nu^2)$.
\item  $\mathcal{S}^+(\textrm{VII}_0)\,\supset\,\mathcal{P}_{W(\nu_1)}\equiv \lim_{\kappa,\nu\rightarrow 0}W(\kappa,\nu_1,\nu^2)$.
\item $\mathcal{C}^+(\textrm{VI}_h)\,\supset\,\mathcal{P}_{W(\kappa,\nu^2)}\equiv \lim_{\nu_1\rightarrow 0}W(\kappa,\nu_1,\nu^2)$.
\item $\mathcal{S}^+(\textrm{VI}_0)\,\supset\,\mathcal{P}_{W(\nu^2)}\equiv \lim_{\kappa,\nu_1\rightarrow 0}W(\kappa,\nu_1,\nu^2)$.
\end{itemize}

\paragraph{In following sections} we will treat the Bianchi sets separately, and the subsets of the Plane Waves and Wonderland belonging therein. As we shall see, there is an anisotropic attractor in all the invariant Bianchi sets we consider.

\newpage
\section*{Dynamical systems analysis}
In the rest of the paper, we perform a dynamical systems analysis of each of the Bianchi sets $\mathcal{B}$(VII$_0$), $\mathcal{B}$(VII$_h$), $\mathcal{B}$(IV) and $\mathcal{B}$(II) separately. In doing a dynamical systems analysis of one of the Bianchi sets, we must consider its closure, since a past or future attractor might be on the boundary of the invariant set.  For those equilibrium points where the remaining constraint is singular, we follow \cite{hewitt93}, performing the analysis in the extended state space instead of the physical part.
\section{Analysing the set $\mathcal{B}$(VII$_h$)}
\label{Sec:typeVIIh}
The closure of $\mathcal{B}$(VII$_h$) is
\begin{equation}
\overline{\mathcal{B}(\textrm{VII}_h)}=\mathcal{B}(\textrm{VII}_h)\cup\mathcal{B}(\textrm{VII}_0)\cup\mathcal{B}(\textrm{V})\cup\mathcal{B}(\textrm{IV})\cup\mathcal{C}(\textrm{II})\cup\mathcal{C}(\textrm{I}).
\end{equation}
As we see, equilibrium sets from many other Bianchi invaraint sets are expected. Some equilibrium sets were analysed in $N_-$\,-\,gauge (Sec. \ref{Sec:DynSysN}) and others in $F$\,-\,gauge (Sec. \ref{Sec:DynSysF}), all according to what we found easiest to implement for each particular equilibrium set.The Tables \ref{tab:VIIhNm} and \ref{tab:VIIhGg} provide an overview.
\begin{table}[H]
	\centering
	\resizebox{\textwidth}{!}{\begin{tabular}{llccccccccccc}
			\toprule
			\multicolumn{13}{c}{\textbf{Equilibrium sets in $\overline{\mathcal{B}(\textrm{VII}_h})$ analysed in $N_-$\,-\, gauge.}} \\
			\hline
Set&$\mathcal{P}$ & $q$ & $\gamma $&$\alpha^2$&$A$&$\Omega_{\rm pf}$&$\Sigma_+$&$\Sigma_-$&$\Sigma_\times$&$N_-$&$\Theta$&$V_1$\\
\hline
$\mathcal{C}^0$(II)& CS & $-1+\frac{3}{2}\gamma$& $(\frac{2}{3},2)$&$1$&0&$\frac{3}{16}(6-\gamma)$&$-\frac{3}{16}(\gamma-\frac{2}{3})$&$\pm\sqrt{3}\frac{3}{16}(\gamma-\frac{2}{3})$&0&$\pm\frac{3}{8}\sqrt{(2-\gamma)(\gamma-\frac{2}{3})}$&0&0\\
$\mathcal{C}^+$(VII$_h$)&	PW$(\alpha,\beta_1,\nu^2)$& $-2\beta_1$& $[0,2]$&$>1$&$1+\beta_1$&0&$\beta_1\,\leq\,0$&0&$-\nu$& $\nu$&$-V_1$&$\pm\sqrt{-\beta_1(1+\beta_1)-\nu^2}$ \\
{\color{gray}$\mathcal{S}^0$(V)}&{\color{gray}M}&{\color{gray} 0}& {\color{gray}$[0,2]$}&{\color{gray}free}&{\color{gray}$1$}&{\color{gray}0}&{\color{gray}0}&{\color{gray}0}&{\color{gray}0}&{\color{gray}0}&{\color{gray}0}&{\color{gray}0}\\
\hline

	\end{tabular}}
	\caption{Summary of equilibrium sets $\mathcal{P}$ analyzed in $N_-$\,-\,gauge, where $N_+=\sqrt{3}\alpha N_-$. Here $\beta_1\,>\,-1$. M is per definition part of PW($\alpha, \beta_1,\nu^2$), and is therefore shadow-faced.} 
	\label{tab:VIIhNm}
\end{table}
\begin{table}[H]
	\centering
	\resizebox{\textwidth}{!}{\begin{tabular}{llccccccccccc}
			\toprule
			\multicolumn{13}{c}{\textbf{Equilibrium sets of $\overline{\mathcal{B}(\textrm{VII}_h})$ analysed in $F$\,-\,gauge.}} \\
			\hline
Set&$\mathcal{P}$ & $q$ & $\gamma $&h&$A$&$\Omega_{\rm pf}$&$\Sigma_+$&$\Sigma_-$&$\Sigma_\times$&$N_+$&$\Theta$&$V_1$\\
\hline
$\mathcal{S}^0$(I)& flat FLRW & $-1+\frac{3}{2}\gamma$& $[0,2]$&undef.&$0$&1&0&0&0&0&0&0 \\
$\mathcal{S}^0$(V)& open FLRW & 0& $\frac{2}{3}$&$h\rightarrow\infty$&$A\,\in\,[0,1]$&$1-A^2$&0&0&0&0&0&0\\
$\mathcal{C}^0$(I)&JED$(\beta_1,\beta_2,\beta_3)$&2&$[0,2)$&undef.&0&0&$\beta_1$&$\beta_2$&$\beta_3$&0&$[-\sqrt{1-\beta^2},\sqrt{1-\beta^2}]$&0\\
$\mathcal{C}^0$(I)&K$(\beta_1,\beta_2)$&2&$[0,2)$&undef.&0&0&$\beta_1$&$\beta_2$&$[-\sqrt{1-\beta_1^2-\beta_2^2},\sqrt{1-\beta_1^2-\beta_2^2}]$&0&0&0\\
$\mathcal{C}^0$(I)&JS$(\beta_1,\beta_2,\beta_3,\Theta)$&2&$2$&undef.&0&$\sqrt{1-\beta^2-\Theta^2}$&$\beta_1$&$\beta_2$&$\beta_3$&0&$\Theta$&0\\
$\mathcal{S}^+$(VII$_h$)&	W($\kappa,\nu_1$)& $-1+\frac{3}{2}\gamma$& $(\frac{2}{3},2)$&$h>0$&$-\frac{3}{4}(2-\gamma)\kappa$&$\frac{3}{4}(2-\gamma)(1-\kappa^2)$&$ \frac{1}{2} -\frac{3}{4}\gamma$&0&0&$\nu_1$& $\kappa\,V_1$&$\mp\frac{3}{4} \sqrt{(2-\gamma ) (\gamma -\frac{2}{3})}$ \\
\hline
	\end{tabular}}
	\caption{Summary of equilibrium sets $\mathcal{P}$ analyzed in $F$\,-\,gauge. In all the Equilibrium sets above, $\nu^2=0$. Notation is such that $\beta^2\,\equiv\,\beta_1^2+\beta_2^2+\beta_3^2$. The parameter $\kappa$ is restricted according to $-1<\,\kappa\,\leq 0$. The group parameter $h$ has to be positive in $\mathcal{B}$(VII$_h$).}
\label{tab:VIIhGg}
\end{table}
One may wonder, perhaps, why the two equilibrium sets Edge and Rope do not show up in this analysis. After all, they are situated in $\mathcal{B}$(I). Not, however, in the part of $\mathcal{B}$(I) that is included in the boundary of $\mathcal{B}$(VII$_h$); namely $\mathcal{C}$(I). Recall that the constraint analysis resulted in $\mathbf{V}_{\rm c}=\mathbf{\Sigma}_1=0$ for $\mathcal{B}$(VII$_h$), but not for $\mathcal{B}$(I) in general. The rotating vectors are in $\mathcal{D}^+$(I)$\,\subset\,\mathcal{B}$(I), and not reachable from $\mathcal{B}$(VII$_h$). These matters have been thoroughly investigated~\cite{thorsrud19}.
	
\subsection{Discussion of stability}
The eigenvalues around each equilibrium point reveals the local stability. In \ref{App:Eigenvalues} the two tables \ref{tab:EigVIIhNm}  and \ref{tab:EigVIIhGg} provide the eigenvalues of the equilibrium sets found in $\mathcal{B}$(VII$_h$). In the following we use these tables to determine the local stability of each equilibrium set.

\paragraph{FLRW: }The no-hair theorem \ref{NoHair} states that \textit{flat} FLRW is the \textit{global attractor} for $0\,<\,\gamma\,<\,2/3$. The local stability analysis confirms this: flat FLRW is \textit{stable} in the extended state space for $\gamma<2/3$.
With (five) unexplained zero-eigenvalues, center-manifold analysis is required for the \textit{open} FLRW model. The remaining eigenvalues are all negative, so open FLRW  is either an attractor or a saddle.\\

\paragraph{Collins-Stewart, CS:} The Collins-Stewart equilibrium point has two positive and four negative eigenvalues, and is therefore a saddle. This remains true also if any one of the eigenvalues is removed. Hence CS will be a saddle also in $\mathcal{B}$(IV).  From the symmetries \eref{NSym2}-\eref{NSym3} of the dynamical system (Sec. \ref{Sec:DynSysN}), we may conclude that CS comes in four copies ($\pm\,\abs{(N_-)_{\rm CS}},(\Sigma_-)_{\rm CS}<0,\alpha=1$ and $\pm\,\abs{(N_-)_{\rm CS}},(\Sigma_-)_{\rm CS}>0,\alpha=-1$) with the same stability.

\paragraph{Plane Waves, PW($\alpha,\beta_1,\nu$):} For $\beta_1\,>\,-\frac{3}{4}\left( \gamma-\frac{2}{3}\right)$, one finds that all the eight eigenvectors in the physical state space are negative. Hence, for this parameter range, the equilibrium set is an attractor. The analysis is inconclusive for $\beta_1\,=\,-\frac{3}{4}\left( \gamma-\frac{2}{3}\right)$, where there is one eigenvalue too many. For $\beta_1\,<\,-\frac{3}{4}\left( \gamma-\frac{2}{3}\right)$, it is a saddle, with a one-dimensional unstable manifold. Note that PW($\alpha,\beta_1,\nu$) comes in two copies. The symmetry \eref{FSym1} of the dynamical system (Sec. \ref{Sec:DynSysN}) ensures that the two copies have the same stability. Note that the local stability analysis is insensitive to the value of $\alpha$. This means that the conclusion will be the same also in $\mathcal{B}$(IV), where $\alpha^2=1$, and in $\mathcal{B}$(VI$_h$), where $\alpha^2<1$. The same must be true for the Milne subset below. 

\paragraph{Milne, M:} The Milne equilibrium set is the subset of the PW equilibrium set where $(\beta_1,\nu)\rightarrow (0,0)$. This gives three zero-eigenvalues, which correspond to the three parameters of the PW set. From table \ref{tab:EigVIIhNm} we see that Milne therefore is an attractor for $\gamma>2/3$ and a saddle for $\gamma<2/3$. center-manifold analysis is required for the point $\gamma =2/3$. 

\paragraph{Wonderland, W($\kappa,\nu_1$)}: To analyse the Wonderland fabric, we note that the linearisation matrix around the equilibrium set consists of two block diagonal matrices $M_1$ and $M_2$. As more thoroughly explained in \ref{App:WonderlandStab}, the matrix $M_2$ consists of $N_+\,\times\,$($\mathcal{B}$(V))$\backslash\Sigma_-$. For this matrix we should therefore expect the same results as was obtained in the analysis of the $\mathcal{B}$(V) subsystem in Section 9.4 of \cite{normann18} (see Table 4 therein). The conclusion of the local analysis there was that Wonderland is an attractor for all parameter values. We find the same here, alongside an extra zero-eigenvalue along $N_+$, because of the group parameter $h$. Also, $M_1$ has eigenvalues with negative real part for all parameter values except $\nu_1\,=\,0$, which is a junction point with the Wonderland in $\mathcal{B}$(VI), W($\kappa,\nu^2$). Hence, Wonderland in $\mathcal{B}$(VII$_h$) is an attractor. Refer to the appendix for further details. Also note the invariance  of the dynamical system (Sec. \ref{Sec:DynSysF}) under the transformation \eref{FSym1}. This ensures that different copies of Wonderland have the same stability. Figure \ref{Fig:LRSVIIh} shows phase flow in the $(\Sigma_+,V_1)$\,-\,plane  of the LRS subset $\mathcal{S}$(VII$_h$), where W($\kappa,\nu_1$) lies, for some value of $\nu_1$. The phase flow is the same as that of the set $\mathcal{B}$(V), displayed in Fig. 1 in \cite{normann18}. The semi circle in our figure shows the path along which the equilibrium set $W(\nu_1)$ will move from FLRW ($\gamma=2/3$) towards $K_-$ ($\gamma=2$) as a function of $\gamma.$

\paragraph{Jacobs' Extended Disk, JED ($\beta_1,\beta_2,\beta_3$):} Analysing JED in the physical state space we find that it is a repeller for $\beta_1\,>\,\sqrt{3}\sqrt{\beta_2^2+\beta_3^2}-1$. If $\beta_1\,<\,\sqrt{3}\sqrt{\beta_2^2+\beta_3^2}-1$, then JED is a saddle. For the remaining case where $\beta_1\,=\,\sqrt{3}\sqrt{\beta_2^2+\beta_3^2}-1$ the analysis is inconclusive. Specifying the JED family to the Kasner subset ($\Theta=0\,\rightarrow\,\beta_3^2=1-\beta_1^2-\beta_2^2$) one finds that K is a repeller if $\beta_1>1/2$ and a saddle if $-1<\beta_1<1/2$. Otherwise, if $\beta_1=-1\,\cup\,1/2$, the analysis is inconclusive. By such, the three zero-eigenvalues of Kasner are explained by K being part of a three-parameter equilibrium set. The analysis holds regardless of the value of $h$, and hence it is valid for $\mathcal{B}$(VI$_h$) also.
\paragraph{Jacobs' Sphere, JS($\beta_1,\beta_2,\beta_3,\Theta$): } The eigenvalues are the same as in the JED case, except for $3(2-\gamma)$, which is now $0$, since $\gamma=2$. This corresponds to the extra parameter $\Theta$ compared to JED. The stability categories, however, must be the same as for JED.

\paragraph{}Table \ref{tab:StabVIIhNm} summarizes the overall stability of the equilibrium sets found in $\mathcal{B}$(VII$_h$).
\begin{table}[H]
	\centering
	\resizebox{\textwidth}{!}{\begin{tabular}{lcclclclclc}
			\toprule
			\multicolumn{11}{c}{\textbf{Classification of equilibrium sets in $\overline{\mathcal{B}(\textrm{VII}_h})$}} \\
			\hline
$\mathcal{P}$&& Existence && Attractor && Saddle && Repeller && Inconclusive \\
			\hline
PW($\alpha,\beta,\nu^2$) && $\gamma\in[0,2]$ &&$\beta_1\,>\,-\frac{3}{4}\left( \gamma-\frac{2}{3}\right)$&&$\beta_1\,<\,-\frac{3}{4}\left( \gamma-\frac{2}{3}\right)$&&&&$\beta_1\,=\,-\frac{3}{4}\left( \gamma-\frac{2}{3}\right)$\\
W($\kappa,\nu_1$)&& $\gamma\in(\frac{2}{3},2)$&& $\forall\, \kappa,\gamma$  && &&&& \\
open FLRW  && $\gamma = \frac{2}{3}$ &&&&&&&&$\forall$\\
flat FLRW && $\gamma\in[0,2)$ &&$\gamma\,\in\,[0,\frac{2}{3})$&&$\gamma\,\in\,(\frac{2}{3},2)$&&&&$\gamma=\frac{2}{3}$\\
K$(\beta_1,\beta_2)$&& $\gamma\in[0,2)$ &&&&\textrm{else}&&$\beta_1>\frac{1}{2}$&&$\beta_1=-1\,\cup\,\frac{1}{2}$\\
JED$(\beta_1,\beta_2,\beta_3)$&& $\gamma\in[0,2)$ &&&&\textrm{else}&&$\beta_1>-1+\sqrt{3}\sqrt{\beta_2^2+\beta_3^2}$&&$\beta_1=-1+\sqrt{3}\sqrt{\beta_2^2+\beta_3^2}$\\
JS$(\beta_1,\beta_2,\beta_3,\Theta)$&& $\gamma=2$ &&&&\textrm{else}&&$\beta_1>-1+\sqrt{3}\sqrt{\beta_2^2+\beta_3^2}$&&$\beta_1=-1+\sqrt{3}\sqrt{\beta_2^2+\beta_3^2}$\\
CS&&$(\frac{2}{3},2)$&&&&$\forall\,\gamma$&&&&\\\hline
	\end{tabular}}
	\caption{The domains where the local stability analysis is conclusive are divided into attractor, saddle and repeller subdomains.  The rightmost column shows the domains where the linear stability analysis is inconclusive. Refer to the text for details regarding the classification of PW and W. }
	\label{tab:StabVIIhNm}
\end{table} 

\begin{figure}[t!]
	\centering
	\includegraphics[width=\textwidth]{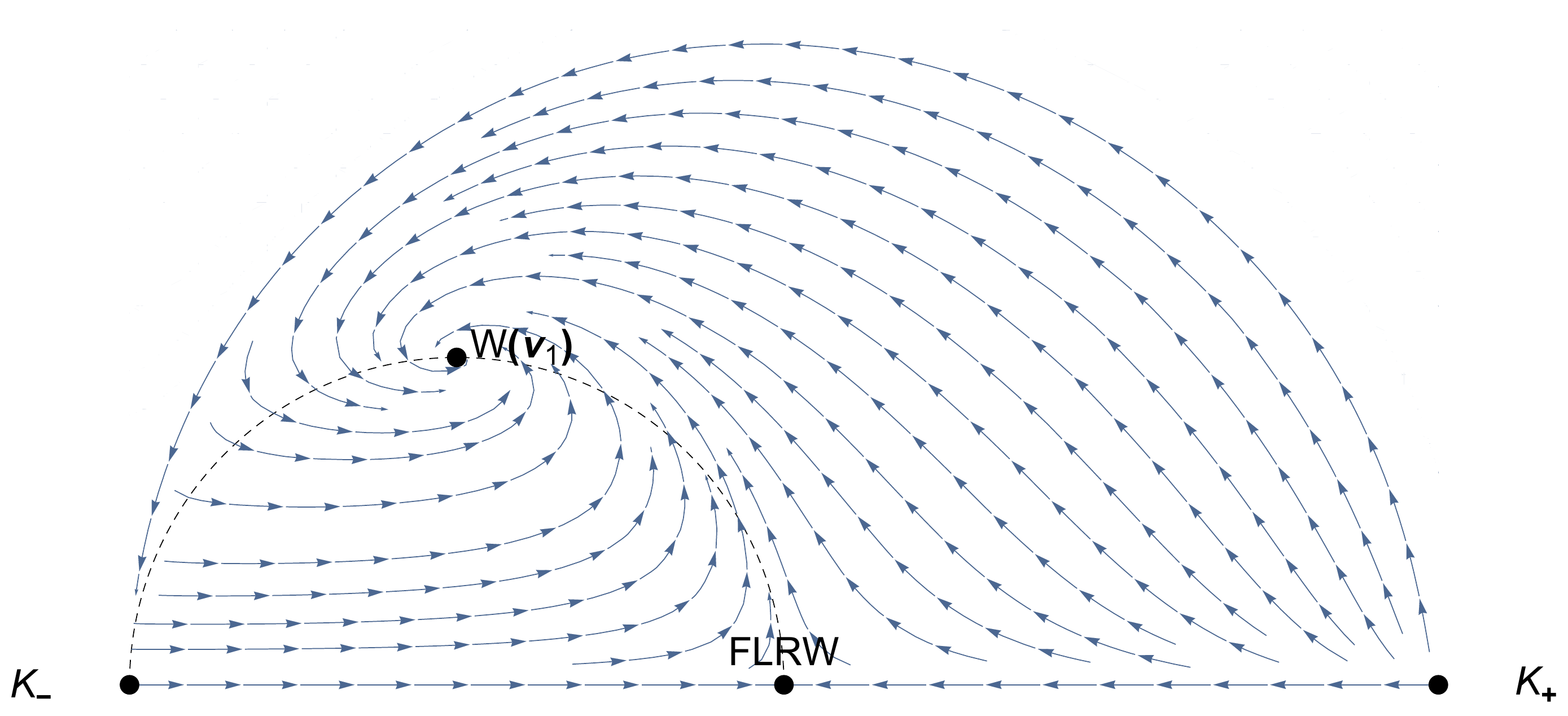}\caption{Phase flow in the LRS subsystem of $\mathcal{B}$(VII$_h$), where $\Sigma_-=\Sigma_\times=\Sigma_3=0$. Here with $\gamma=4/3$. ``W'' denotes Wonderland and $K_\pm$ denotes Kasner solutions with $\Sigma_+= \pm1$. The dashed semicircle shows the circle on which W$(\kappa,\nu_1$) moves as a function of $\gamma$. }
	\label{Fig:LRSVIIh}
\end{figure}


\section{Analysing the set $\mathcal{B}$(VII$_0$)}
\label{Sec:VII0}
The closure of $\mathcal{B}$(VII$_0$) is
\begin{equation}
\overline{\mathcal{B}(\textrm{VII}_0)}=\mathcal{B}(\textrm{VII}_0)\cup\mathcal{C}(\textrm{II})\cup\mathcal{C}(\textrm{I}).
\end{equation}
Table \ref{tab:VII0Nm0} gives an overview of the equilibrium sets found in $\overline{\mathcal{B}(\textrm{VII}_0)}$ in $F$\,-\,gauge. With $A=0$ it is evident from the equations that the timelike part of the $j$-form field will vanish asymptotically (except for $q=2$). The spatial part of the form field, however, which is the part responsible for the anisotropies, does not necessarily die away. The dynamical system (Sec. \ref{Sec:DynSysF}) is 7-dimensional (since $A=0$). 
\begin{table}[H]
	\centering
	\resizebox{\textwidth}{!}{\begin{tabular}{llccccccccccc}
			\toprule
			\multicolumn{13}{c}{\textbf{Equilibrium sets of $\overline{\mathcal{B}(\textrm{VII}_0})$ analysed in  $F$\,-\,gauge}} \\
			\hline
Set&$\mathcal{P}$& $q$ & $\gamma $&h&$\Omega_{\rm pf}$&$\Sigma_+$&$\Sigma_-$&$\Sigma_\times$&$N_+$&$N_\times$&$\Theta$&$V_1$\\
\hline
$\mathcal{S}^0$(I)&flat FLRW&$-1+\frac{3}{2}\gamma$&$[0,2)$&undef.&1&0&0&0&0&0&0&0\\
$\mathcal{C}^0$(I)&JS($\beta_1,\beta_2,\beta_3,\Theta$)&2&2&undef.&$\sqrt{1-\Theta^2-\beta^2}$&$\beta_1\,\in\,[-1,1]$&$\beta_2\,\in\,[-\sqrt{1-\beta_1^2},\sqrt{1-\beta_1^2}]$&$\beta_3\,\in\,[-\sqrt{1-\beta_1^2-\beta_2^2},\sqrt{1-\beta_1^2-\beta_2^2}]$&0&0&$[-\sqrt{1-\beta^2},\sqrt{1-\beta^2}]$&0\\
$\mathcal{C}^0$(I)&JED($\beta_1,\beta_2,\beta_3$)&2&$[0,2)$&undef.&0&$\beta_1\,\in\,[-1,1]$&$\beta_2\,\in\,[-\sqrt{1-\beta_1^2},\sqrt{1-\beta_1^2}]$&$\beta_3\,\in\,[-\sqrt{1-\beta_1^2-\beta_2^2},\sqrt{1-\beta_1^2-\beta_2^2}]$&0&0&$[-\sqrt{1-\beta^2},\sqrt{1-\beta^2}]$&0\\
$\mathcal{C}^0$(II)&	CS& $-1+\frac{3}{2}\gamma$& $(\frac{2}{3},2)$&undef.&$\frac{3}{16} (-\gamma +6)$&$\frac{3}{16} \left(\frac{2}{3}-\gamma \right)$&0&$\sqrt{3}\frac{3}{16}\left(\frac{2}{3}-\gamma\right)$&$\pm\sqrt{3}\frac{3}{8} \sqrt{(2-\gamma ) \left(\gamma -\frac{2}{3}\right)} $&$\pm\frac{3}{8} \sqrt{(2-\gamma ) \left(\gamma -\frac{2}{3}\right)} $&0&0\\
$\mathcal{S}^+$(VII$_0$)&	W($\nu_1$)& $-1+\frac{3}{2}\gamma$& $(\frac{2}{3},2)$&$0$&$\frac{3}{2}-\frac{3 \gamma }{4}$&$ \frac{1}{2} - \frac{3}{4}\gamma$&0&0&$\nu_1$& 0&0&$\mp\frac{3}{4} \sqrt{(2-\gamma ) (\gamma -\frac{2}{3})}$\\
\hline
	\end{tabular}}
	\caption{Summary of equilibrium sets analyzed in $F$\,-\,gauge.  $N_-=0$ for all the equilibrium sets.  Here notation is such that $\beta^2\,\equiv\,\beta_1^2+\beta_2^2+\beta_3^2$. Keep in mind that every solution has to fulfill the Hamiltonian constraint. Note that $h=0$ in $\mathcal{B}(\textrm{VII}_0)$.}
	\label{tab:VII0Nm0}
\end{table}
\subsection{Discussion of stability}
\paragraph{Local stability.} The local stability analysis is very similar to that of $\mathcal{B}(\textrm{VII}_h)$, as the only difference is that we now have the restriction $A=0$. Table \ref{tab:EigVII0R} gives the eigenvalues of the equilibrium sets computed in the $F$\,-\,gauge, and Table \ref{tab:StabVII0} summarizes the local stability.

\paragraph{Global Stability:} By use of the monotonic function $Z_6$  of \ref{App:B}, re-expressed here as $Z_6^\prime=\phi^{-1}Z_6F$,  where $F$ is the expression inside the square brackets in eq.\eref{b2}, it is possible to give results on the global stability. We note that $\phi$ in the expression for the monotonic function is always positive, except if $\gamma=2$ and $\Sigma_+=-1$. We therefore reach the following global conclusions:
\begin{itemize}
\item $0\,<\gamma\,<\frac{2}{3}$: The no-hair theorem\ref{NoHair} shows that the global future attractor is the  flat FLRW space-time. As a past attractor, the only option is $Z_6(\tau\rightarrow\,-\infty)\rightarrow\,0$, since there are no candidates satisfying $F=0$. $Z_6(\tau\rightarrow\,-\infty)\rightarrow\,0$ is possible if $\Omega\rightarrow 0$, which identifies JED as the global past attractor.
\item $\gamma=\frac{2}{3}$: For this value of $\gamma$ there are no future attractor candidates with $\Omega=0$. Hence the only available option is now $F=0$, which in this case implies $\Sigma_+=\Sigma_-=\Sigma_\times=\Theta=0$, and again flat FLRW is the future attractor. Looking at the repellers, we see that there are no past attractor candidates satisfying $F=0$. Hence, in this case we must require $\Omega=0$ and again JED is the global past attractor.
\item $\frac{2}{3}\,<\gamma\,< 2$: In this case the future attractor is found by requiring $F=0$, which is the only option, judging from the table. This uniquely identifies Wonderland as global future attractor. The only valid option for $Z_6^\prime=0$ in the past, is $\Omega=0$. The global past attractor must therefore again be JED.
\item $\gamma=2$: In this case, only JS is found in the table. The global past attractor is therefore JS with $\beta_1>\sqrt{3}\sqrt{\beta_2^2+\beta_3^2}-1$ (or possibly also including $\beta_1=\sqrt{3}\sqrt{\beta_2^2+\beta_3^2}-1$). The only option for a global future attractor is JS with $\Sigma_+=-1$, where the first order stability analysis breaks down. This identifies one of the Taub points, namely $T_1$ (the Taub form of flat space-time, cf.~\cite[Sec. 9.1.6]{dynSys}) where $\Sigma_+=-1$, as the global future attractor.
\end{itemize} 

Table \ref{tab:StabVII0} shows the overall stability of the equilibrium sets found in $\overline{\mathcal{B}(\textrm{VII}_0})$.
\begin{table}[h]
	\centering
	\resizebox{\textwidth}{!}{\begin{tabular}{lcclclclclc}
			\toprule
			\multicolumn{11}{c}{\textbf{Classification of equilibrium sets in $\overline{\mathcal{B}(\textrm{VII}_0})$}} \\
			\hline
$\mathcal{P}$&& Existence && Attractor && Saddle && Repeller && Inconclusive \\
\hline
W($\nu_1$)&& $\gamma\in(\frac{2}{3},2)$ && $\forall\,\gamma,\,\nu_1$  && &&&& \\
flat FLRW && $\gamma\in[0,2)$ &&$\gamma\,\in\,[0,\frac{2}{3})$&&$\gamma\,\in\,(\frac{2}{3},2)$&&&&$\gamma=\frac{2}{3}$\\
K$(\beta_1,\beta_2)$ && $\gamma\in[0,2)$ &&&&\textrm{else}&&$\beta_1>\frac{1}{2}$&&$\beta_1=-1\,\cup\,\frac{1}{2}$\\
JED$(\beta_1,\beta_2,\beta_3)$&& $\gamma\in[0,2)$ &&&&\textrm{else}&&$\beta_1>-1+\sqrt{3}\sqrt{\beta_2^2+\beta_3^2}$&&$\beta_1=-1+\sqrt{3}\sqrt{\beta_2^2+\beta_3^2}$\\
JS$(\beta_1,\beta_2,\beta_3,\Theta)$&& $\gamma=2$ &&&&\textrm{else}&&$\beta_1>-1+\sqrt{3}\sqrt{\beta_2^2+\beta_3^2}$&&$\beta_1=-1+\sqrt{3}\sqrt{\beta_2^2+\beta_3^2}$\\
CS&& $\gamma\in(\frac{2}{3},2)$ &&&&$\forall\,\gamma$ &&&&  \\
\hline
\end{tabular}}
	\caption{The domains where the stability analysis is conclusive are divided into attractor, saddle and repeller subdomains. The rightmost column shows the domains where the linear stability analysis is inconclusive.}
	\label{tab:StabVII0}
\end{table} 
\newpage
\section{Dynamical system for $\mathcal{B}$(II) and $\mathcal{B}$(IV)}
\subsection{Detailed discussion of constraints}
Earlier we discussed the constraints rather generally. In the following we specify to the particular sets $\mathcal{B}$(II) and $\mathcal{B}$(IV), giving a more detailed discussion of the two constraints $C_1$ and $C_2$,  equations \eref{Constr1} and \eref{Constr4} respectively.

\paragraph{The constraint $C_1$.} As already discussed in a previous section, this constraint gives for $\mathcal{B}$(IV) the only option $\mathbf{V}=V^1\mathbf{e}_1$. In $\mathcal{B}$(II)  however, this constraint is identically fulfilled\setcounter{footnote}{0}\footnote{With $A=0$, the constraint reads $
\sqrt{3}\mathbf{N}_\Delta^*\mathbf{V}_c   +N_+\mathbf{V}_c^*=0$ (alongside the complex conjugate eq.). Rewriting to real, scalar form one readily verifies that this eq. is identically fulfilled in $\mathcal{B}$(II), where $3(N_-^2+N_\times^2)-N_+^2=0$.}. Since the matrix $N_{ab}$ only has one non-zero eigenvalue in $\mathcal{B}$(II), and since $A=0$, we need another vector along which to align our frame without ambiguity. Equivalently: We have  got extra gauge freedom. We use this freedom to align our frame along $\mathbf{V}$ instead. By such we may still conveniently choose $\mathbf{V}=V^1\mathbf{e}_1$. Henceforth we find
\begin{eqnarray}
\mathbf{V}_{\rm c}=0,\quad\cases{\textrm{only option in $\mathcal{B}$(IV),}\\
\textrm{gauge choice in $\mathcal{B}$(II).}}
\end{eqnarray}
In both cases there is still gauge freedom left: The unspecified variables are $\{\phi_1,R_1\}$.

\paragraph{The second constraint $C_2$} is strategically  simplified by $\mathbf{V}_c=0$. As previously discussed, there is for the $\mathcal{B}$(IV) no choice but $\mathbf{\Sigma}_{\rm 1}=0 $. For $\mathcal{B}$(II), the situation is again more delicate. Going to real variables, the constraint takes the form as follows.
\begin{eqnarray}
&(\sqrt{3}N_-+N_+)\Sigma_2+(\sqrt{3}N_\times-3A)\Sigma_3=0\\
&(\sqrt{3}N_--N_+)\Sigma_3-(\sqrt{3}N_\times+3A)\Sigma_2=0.
\end{eqnarray}
Implementing next the $\mathcal{B}$(II) specifications $3(N_-^2+N_\times^2)-N_+^2=0$ and $A=0$ we find from the above set of equations that only one of them is non-identically satisfied at the time. This conclusion is reached by substituting away $\Sigma_2$ or $\Sigma_3$ in one equation from the other. The two options we thus have are 
\begin{eqnarray}\fl
(\sqrt{3}N_-+N_+)\Sigma_2=-\sqrt{3}N_\times\Sigma_3=0\quad\quad\textrm{or}\quad\quad(\sqrt{3}N_--N_+)\Sigma_3=\sqrt{3}N_\times\Sigma_2=0.
\end{eqnarray}

Actually, either of these options are well studied in the $N_-$\,-\,gauge (see Sec. \ref{Subsec:Gauge}). Using the remaining gauge freedom $\{\phi_1,R_1\}$, we specify to $N_\times=0$ and $N_+=\sqrt{3}\alpha N_-$. Including $A$, the two options we have become
\begin{eqnarray}
\fl\sqrt{3}N_-(1+\alpha)\Sigma_2-3A\Sigma_3=0\quad\quad\textrm{or}\quad\quad\sqrt{3}N_-(1-\alpha)\Sigma_3-3A\Sigma_2=0.
\end{eqnarray}
To be in $\mathcal{B}$(II) or $\mathcal{B}$(IV) one must have $N_-\,\neq\,0$ and $\alpha=\pm\,1$. In summary we are therefore left with options as follows:
 \begin{eqnarray}
&\fl\label{o1}\mathcal{B}\textrm{(II):}\quad\quad\quad &(\alpha,\Sigma_2,\Sigma_3)\rightarrow (-1,0,\Sigma_3)\,\cup\,(1,\Sigma_2,0).\\
&\label{o2}\fl\mathcal{B}\textrm{(IV):}\quad\quad\quad & (\alpha,\Sigma_2,\Sigma_3)\rightarrow (\pm 1,0,0).
 \end{eqnarray}
At this point all gauge freedom is used. It is in place with a final comment about the case $N_-=0$, where the parameterization breaks down. $N_-=0$ corresponds to a $\mathcal{B}$(V) subset of $\overline{\mathcal{B}\textrm{(IV)}}$ and a $\mathcal{B}$(I) subset of $\overline{\mathcal{B}\textrm{(II)}}$. Both $\mathcal{B}$(I) and $\mathcal{B}$(V) were analysed in \cite{normann18}. 

\subsection{System of equations in $N_-$\,-\,gauge}
\label{Subsec:B24SystEq}
Based on the above analysis of the constraints, we find the following set of equations, from which we will specify either to $\mathcal{B}$(II) or $\mathcal{B}$(IV). 
\begin{eqnarray}
&\fl\textit{j}\textrm{-form eq.s }\quad\quad\cases{\label{FluidEqsT24}
V_1^\prime=\left(q+2\Sigma_+\right)V_1,\\
\Theta^\prime=(q-2)\Theta-2AV_1\,,
}&\\
&\fl\textrm{Einst. eq.s }\quad\phantom{00}\cases{\label{EinstEqT24}
\Sigma_-^\prime=(q-2)\Sigma_-+2\sqrt{3}\alpha(\Sigma_\times^2 -N_-^2)+\sqrt{3}\left(\Sigma_2^2-\Sigma_3^2\right)\,,\\
\Sigma_\times^\prime=(q-2-2\sqrt{3}\alpha\Sigma_-)\Sigma_\times-2A N_- +2\sqrt{3}\Sigma_2\Sigma_3\,,\\
\Sigma_2^\prime=(q-2-3\Sigma_+-\sqrt{3}\Sigma_-)\Sigma_2+\sqrt{3}\Sigma_\times\Sigma_3(\alpha-1)\,,\\
\Sigma_3^\prime=(q-2-3\Sigma_++\sqrt{3}\Sigma_-)\Sigma_3-\sqrt{3}\Sigma_\times\Sigma_2(\alpha+1)\,,\\
\Sigma_+^\prime=\left(q-2\right)\Sigma_+-2\left(N_-^2+V_1^2\right)+3\left(\Sigma_2^2+\Sigma_3^2\right)\,,
}&\\
&\fl\textrm{En. cons. }\quad\quad\phantom{0}
\cases{\label{EnConsT24}
\Omega_{\rm pf}^\prime=2\left(q+1-\frac{3}{2}\gamma \right)\Omega_{\rm pf}\,,
}&\\
&\fl\textrm{Jacobi id. }\quad\quad\phantom{.}
\cases{\label{JacId2T24}
N_-^\prime=\left(q+2\Sigma_++2\sqrt{3}\alpha\Sigma_-\right)N_-\,,\\
\alpha^\prime=2\sqrt{3}(1-\alpha^2)\Sigma_-\,,\\
A^\prime=\left(q+2\Sigma_+\right)A.
}&
\end{eqnarray}
These dynamical equations are subject to the two remaining constraints, Eq.s. \eref{Constr1} and \eref{Constr2}, which now read
\begin{eqnarray}
&\fl \label{Constr2T24} C_2=1-\Omega_{\rm pf}-\Sigma_+^2-\Sigma_-^2-\Sigma_\times^2-\Sigma_2^2-\Sigma_3^2-\Theta^2-V_1^2-A^2-N_-^2=0\,,\\
&\fl\label{Constr3T24} C_1=\Theta V_1-A\Sigma_++N_-\Sigma_\times=0.
\end{eqnarray}

\paragraph{Symmetry: } Note that the above system of equations is symmetric under
\begin{equation}\fl
(\alpha^\prime,\alpha,\Sigma_-^\prime,\Sigma_-,\Sigma_2^\prime,\Sigma_2,\Sigma_3^\prime,\Sigma_3)\quad\rightarrow\quad(-\alpha^\prime,-\alpha,-\Sigma_-^\prime,-\Sigma_-,\Sigma_3^\prime,\Sigma_3,\Sigma_2^\prime,\Sigma_2)
\end{equation}
Remembering the analysis we performed of the constraints, we can utilize this symmetry to study only one of the two options we found for $\mathcal{B}$(II) in Eq. \eref{o1} and for $\mathcal{B}$(IV) in Eq. \eref{o2}. We make the following choices, without loss of generality.
\begin{eqnarray}
&\fl\label{IIspecs}\mathcal{B}\textrm{(II)}&:\,\,\alpha=1\quad\textrm{and}\quad\Sigma_2=0\quad\textrm{and}\quad A=0.\\
&\fl\label{IVspecs}\mathcal{B}\textrm{(IV)}&:\,\,\alpha=1\quad\textrm{and}\quad\Sigma_2=\Sigma_3=0.
\end{eqnarray}
\newpage
\section{Analysing $\mathcal{B}$(IV)}
In this section we report the results of analysing the invariant set $\mathcal{B}$(IV). From the type constraint (see Sec. \ref{Subsec:GenDynSys}), we conclude that the closure is 
\begin{equation}
\overline{\mathcal{B}\textrm{(IV)}}=\mathcal{B}\textrm{(IV)}\cup \mathcal{C}\textrm{(I)}\cup \mathcal{C}\textrm{(II)}\cup \mathcal{B}\textrm{(V)}.
\end{equation}
We may therefore find equilibrium points from all of these invariant sets. Table \ref{tab:IVNm} shows the equilibrium sets found in the $N_-$\,-\,gauge. The dynamical system is 6-dimensional. Figure~\ref{subfig:B4} shows some of the equilibrium sets as projected over the past attractor region.
\begin{table}[H]
	\centering
	\resizebox{\textwidth}{!}{\begin{tabular}{llcccccccccc}
			\toprule
			\multicolumn{12}{c}{\textbf{Equilibrium sets of $\mathcal{B}$(IV)} analysed in $N_-$\,-\,gauge} \\
			\hline
Set&$\mathcal{P}$ & $q$ & $\gamma $&$A$&$\Omega_{\rm pf}$&$\Sigma_+$&$\Sigma_-$&$\Sigma_\times$&$N_-$&$\Theta$&$V_1$\\
\hline
$\mathcal{C}^0$(I)&JED($\beta_1,\beta_2$)&2&$[0,2]$&0&0&$\beta_1\,\in\,[-1,1]$&$[-\sqrt{1-\beta_1}<\beta_2<\sqrt{1-\beta_1}]$&0&0&$-\sqrt{1-\beta^2},\sqrt{1-\beta^2}$&0\\
$\mathcal{C}^0$(I)&JS$(\beta_1,\beta_2,\Theta)$&2&$2$&0&$\sqrt{1-\beta^2-\Theta^2}$&$\beta_1$&$\beta_2$&0&0&$\Theta$&0\\
$\mathcal{S}^0$(I)& flat FLRW & $-1+\frac{3}{2}\gamma$& $[0,2]$&0&1&0&0&0&0&0&0 \\
$\mathcal{C}$(II)&CS& $-1+\frac{3}{2}\gamma$& $(\frac{2}{3},2)$&0&$-\frac{3}{16} (\gamma -6)$&$\frac{1}{16} (2-3 \gamma )$&$\frac{1}{16} \sqrt{3} (2-3 \gamma )$&0&$\mp\frac{1}{8} \sqrt{3} \sqrt{(8-3 \gamma ) \gamma -4}$&0&0\\
$\mathcal{C}^+$(IV)&	PW$(\beta_1,\nu^2)$& $-2\beta_1$& $[0,2]$&$1+\beta_1$&0&$\beta_1\,\in\,(-1,0)$&0&$-\nu$& $\nu$&$\sqrt{-\beta_1  (\beta_1 +1)-\nu^2}$&$-\sqrt{-\beta_1  (\beta_1 +1)-\nu^2}$ \\
$\mathcal{S}^0$(V)& open FLRW & $-1+\frac{3}{2}\gamma$& $\frac{2}{3}$&$[-1,1]$&$1-A^2$&0&0&0&0&0&0 \\
$\mathcal{S}^0$(V)&M& 0& $[0,2]$&1&0&0&0&0&0&0 \\
$\mathcal{S}^+$(V)&	W($\kappa$)& $-1+\frac{3}{2}\gamma$& $(\frac{2}{3},2)$&$\frac{3}{4}(2-\gamma)\kappa$&$\frac{3}{4}(2-\gamma)(1-\kappa^2)$&$ \frac{1}{2} -\frac{3}{4}\gamma$&0&0& 0&$\kappa\,V_1$&$\mp\frac{3}{4} \sqrt{(2-\gamma ) (\gamma -\frac{2}{3})}$ \\
\hline
	\end{tabular}}
	\caption{Summary of equilibrium sets of $\mathcal{B}$(IV) analysed in $N_-$\,-\,gauge in. Note that since we are in $N_-$\,-\,gauge, $\nu_2^2=\nu^2$. }
	\label{tab:IVNm}
\end{table}

\subsubsection*{Eigenvalues and stability}
\paragraph{Jacobs' Extended Disk, JED($\beta_1,\beta_2$), and Jacobs' Sphere, JS($\beta_1,\beta_2,\Theta$).} These two equilibrium sets have the same eigenvalues, except for the $\gamma$-dependent one, which is 0 for JS, since $\gamma=2$. The stability regions fall into the same categories for these Eq. sets. In particular they are repellers for $0<\,\sqrt{3}\beta_2\,<\beta_1+1$.

\paragraph{FLRW:} The no-hair theorem \ref{NoHair} states that the spatially flat FLRW is the \textit{global attractor} for $0\,<\,\gamma\,<\,2/3$. The local stability analysis confirms this: the flat FLRW is \textit{stable} in the extended state space for $\gamma<2/3$. The \textit{Open FLRW} branches off the flat equilibrium point at $\gamma=2/3$. It is either stable or a saddle, but center-manifold analysis is required in order to find out.

\paragraph{Plane Waves, PW($\beta_1,\nu^2$), and Milne, M:} In the analysis of VII$_h$ we found that the stability of this set was independent of $\alpha$. We therefore refer the reader back to the more general analysis of $\mathcal{B}$(VII$_h$) in Section \ref{Sec:typeVIIh}.

\paragraph{Wonderland, W($\kappa$): } This is the part of Wonderland found in $\mathcal{B}$(V). In the present treatment we find one zero-eigenvalue that cannot be explained by parameters of the equilibrium set. Recalling that  $\nu_1=0$ corresponds to the origo of the parametrization, we may use the results of the analysis of $\mathcal{B}$(VII$_h$) instead. In \ref{App:WonderlandStab} the $\mathcal{B}$(VII$_h$) version of Wonderland is analysed. Here we show that two extra eigenvalues go to zero in the particular point $\nu_1=0$. These correspond to the parameters $\nu_2,\nu_3$ in Wonderland of $\mathcal{B}$(VI$_h$). Hence all zero-eigenvalues may be accounted for, also in the point $\nu_1=0$. Therefore W($\kappa$) in $\mathcal{B}$(IV) must also be stable. 

\paragraph{Collins-Stewart, CS:} All copies of CS are saddles, as already found and explained in Section \ref{Sec:typeVIIh} on $\mathcal{B}$(VII$_h$). 

\paragraph{}Also, Proposition \ref{prop:genProp} shows that FLRW or Milne is the global future attractor for $\gamma=2/3$.

\paragraph{} Table \ref{tab:StabIV} summarizes the overall stability of the equilibrium sets found in $\mathcal{B}$(IV). The zero-eigenvalues correspond either to parameters of the equilibrium set, or they result from an inconclusive first order stability analysis. In the present case, however, we can account for them by counting parameters of the Equilibrium sets.

\begin{table}[h]
\centering
\resizebox{\textwidth}{!}{\begin{tabular}{lcclclclclc}
\toprule
\multicolumn{11}{c}{\textbf{Classification of equilibrium sets in $\mathcal{B}$(IV)}} \\
\hline
$\mathcal{P}$&& Existence && Attractor && Saddle && Repeller && Inconclusive \\
\hline
flat FLRW && $\gamma\in[0,2)$ &&$\gamma\,\in\,[0,\frac{2}{3})$&&$\gamma\,\in\,(\frac{2}{3},2)$&&&&$\gamma=\frac{2}{3}$\\
JED$(\beta_1,\beta_2)$ &&  $\gamma\in[0,2)$ &&&&\textrm{else}&&$0<\,\sqrt{3}\beta_2\,<\beta_1+1$&&$0<\,\sqrt{3}\beta_2\,=\beta_1+1$\\
JS$(\beta_1,\beta_2,\Theta)$&& $\gamma=2$ &&&&\textrm{else}&&$0<\,\sqrt{3}\beta_2\,<\beta_1+1$&&$0<\,\sqrt{3}\beta_2\,=\beta_1+1$\\
C.-S.&& $\gamma\in(\frac{2}{3},2)$ &&&&$\forall\,\gamma$ &&&&\\
PW$(\beta_1,\nu^2)$ && $\gamma\in[0,2]$ &&$\gamma>\frac{2}{3}-\frac{4}{3}\Sigma_+$&&$\gamma<\frac{2}{3}-\frac{4}{3}\Sigma_+$&&&&$\gamma\,=\,\frac{2}{3}-\frac{4}{3}\Sigma_+$\\
open FLRW &&&&&&&&&&$\forall\,\gamma$\\
W($\kappa$)&& $\gamma\in(\frac{2}{3},2)$ && $\forall\,\gamma$  && &&&&\\
 \hline
	\end{tabular}}
	\caption{The domains where the stability analysis is conclusive are divided into attractor, saddle and repeller subdomains by the conditions above.  The rightmost column shows the domains where the linear stability analysis is inconclusive.}
	\label{tab:StabIV}
\end{table} 
\section{Analysing $\mathcal{B}$(II)}
\label{Sec:typeII}
The closure of this set is
\begin{equation}
\overline{\mathcal{B}\textrm{(II)}}=\mathcal{B}\textrm{(II)}\,\cup\,\mathcal{B}\textrm{(I)}.
\end{equation}
Over the course of the following subsections we will show that the future asymptotes are contained in the boundary $\mathcal{B}$(I). Figure~\ref{subfig:B2} shows some of the equilibrium points as projected over the past attractor region.
\subsection{Asymptotic subspaces}
The constraint equation \eref{Constr3T24} reads
\begin{equation}
\Theta V_1+N_-\Sigma_\times=0
\end{equation}
when the $\mathcal{B}$(II) constraint is satisfied. Also, the $\Theta^\prime$ -equation shows that $\Theta$ will asymptotically tend towards zero for $q<2$. Henceforth, one must either have $q=2$ (which uniquely characterizes JED in the absence of a perfect fluid ($\gamma<2$) and JS in the presence of a perfect fluid ($\gamma=2$), or one finds that any path must asymptotically tend towards $N_-=0$ or $\Sigma_\times=0$. These are invariant subsets.

\begin{lem}[Asymptotic subspaces of $\mathcal{B}$(II)]
\label{lemBII1}
For $q<2$ all orbits with $V_1^2\,>\,0$ in the set $\mathcal{B}$(II) are future asymptotic to $\mathcal{B}$(II)$\vert_{N_-=0}\,\cup\,\mathcal{B}$(II)$\vert_{\Sigma_\times=0}$.
\end{lem}
\begin{proof}
The monotonic function $\Theta$ in conjunction with the constraint  \eref{Constr3T24} with $A=0$. 
\end{proof}
Since we know that the orbits decay into these two invariant subsets, the key to understand the future development of the $\mathcal{B}$(II) set is to understand these two subsets. 
\begin{itemize}
\item $\mathcal{B}$(II)$\vert_{N_-=0}$: This is $\mathcal{B}$(I), and may be reached asymptotically. The resulting dynamical system is identical to the case treated in \cite{normann18}, where W,R and E were established as global future attractors for $\gamma\,>\,2/3$. With the parametrization \eref{param} to be introduced below, the subspace $\mathcal{B}$(II)$\vert_{N_-=0}$ corresponds to $\theta=\left(n+\frac{1}{2}\pi\right)$.
\item $\mathcal{B}$(II)$\vert_{\Sigma_\times=0}$: This is the remaining part of $\lim_{\tau\rightarrow\infty}\mathcal{B}$(II). With the parametrization \eref{param} to be introduced below this corresponds to $\zeta=0$. Note the particular form of the equation for $\Sigma_-^\prime$ in this set. From \ref{Subsec:B24SystEq} with the specifications \eref{IIspecs} and $\Sigma_\times=0$ we find
\begin{equation}
\Sigma_-^\prime=(q-2)\Sigma_--2\sqrt{3}N_-^2-\Sigma_3^2.
\end{equation}
This equation decreases monotonically for $\Sigma_->0$. Hence, the asymptotic value of $\Sigma_-$ in this set must be non-positive. We summarize the findings in the following lemma.
\begin{lem}[Future asymptotes of $\mathcal{B}$(II)$\vert_{\Sigma_\times=0}$] 
\label{lemBII2}The set $\mathcal{B}$(II)$\vert_{\Sigma_\times=0}$ is future asymptotic to $\mathcal{B}$(II)$\vert_{\Sigma_\times=0,\Sigma_-\,\leq\,0.}$
\end{lem}
Refer to the text above for a proof.
\end{itemize}
In establishing the global behaviour of the set $\mathcal{B}$(II) we shall only have to analyse further the part which is not already analysed, namely $\mathcal{B}$(II)$\vert_{\Sigma_\times=0,N_-\,\neq\,0}$. In order to do so, it is useful to perform a local stability analysis, which we do for the whole set $\mathcal{B}$(II) in the following three subsections. But first, note the following theorem.

\begin{thm}[Future asymptote for $\gamma=2/3$]
\label{genThmBII}
The Bianchi invariant set $\mathcal{B}$(II) with a perfect fluid and a $j$-form fluid are future asymptotic to
\begin{equation}
\Omega_{\rm pf}=1.
\end{equation}
\end{thm}
\begin{proof}
We refer to the proof of theorem 8.1 in \cite{normann18}, of which the above lemma is an extension. The proof is similar.
\end{proof}
\subsection{Parametrized dynamical system for $\mathcal{B}$(II)}
In the set $\mathcal{B}$(II), the constraint \eref{Constr3T24} ($A=0$) is well suited for the parametrization
\begin{eqnarray}
\label{param}
&(N_-,V_1)=\eta(\cos\theta,\sin\theta)\\
&(\Theta,\Sigma_\times)=\zeta(\cos\theta,-\sin\theta).
\end{eqnarray}
Rewriting the system \eref{FluidEqsOpt1a}-\eref{JacId2Opt1a} in terms of this parametrization, and the further restrictions for $\mathcal{B}$(II) mentioned earlier, the final system is
\begin{eqnarray}
\label{TypeIIFinal}
&\Sigma_3^\prime=(q-2+\sqrt{3}\Sigma_--3\Sigma_+)\Sigma_3\,,&\\
&\Sigma_-^\prime=(q-2)\Sigma_--\sqrt{3}\Sigma_3^2\,+\,2\sqrt{3}(\zeta^2\sin^2\theta-\eta^2\cos^2\theta)\,,&\\
&\Sigma_+^\prime=\left(q-2\right)\Sigma_++3\Sigma_3^2-2\eta^2\,,&\\
&\eta^\prime=\left(q+2\Sigma_+ +2\sqrt{3}\Sigma_-\cos^2\theta\right)\eta \,,&\\
&\zeta^\prime=\left(q-2-2\sqrt{3}\Sigma_-\sin^2\theta\right)\zeta \,,&\\
&\Omega^\prime=2(q+1-\frac{3}{2}\gamma)\Omega\,,&\\
&\theta^\prime=- 2\sqrt{3}\Sigma_-\sin\theta\cos\theta\,,&
\end{eqnarray}
together with the Hamiltonian constraint \eref{Constr2T24}, which now takes the form
\begin{equation}
C_1=1-\Omega_{\rm pf}-\Sigma_3^2-\Sigma_+^2-\Sigma_-^2-\eta^2-\zeta^2.
\end{equation}
Using this to remove one variable (we choose to remove $\Omega_{\rm pf}$), the resulting dynamical system is 6-dimensional. The deceleration parameter \eref{q} takes the form
\begin{equation}
\label{q2}\fl
q=2\Gamma^2+\frac 12({3}\gamma -2)\Omega_{\rm pf},\quad\textrm{where}\quad\Gamma^2\equiv\Sigma_{+}^2+\Sigma_-^2+\Sigma_3^2+\zeta^2,
\end{equation}
and one may note that the above system has the property
\begin{equation}
\left(\Gamma^2\right)^\prime-2(q-2)\Gamma^2=2q\eta^2-\left(\eta^2\right)^\prime.
\end{equation}
This equation proves that with $\eta=\eta^\prime=0$, we must have $\lim_{\tau\rightarrow\infty}\Gamma^2=0$. Hence, this asymptotic subset is future asymptotic to $\Omega_{\rm pf}=1$ (and hence FLRW), and past asymptotic to JED ($\gamma<2$) or JS ($\gamma=2$). This is in agreement with the conclusion reached in our previous work, where we analysed $\mathcal{B}$(I).

\paragraph{Symmetries.} We shall also note that the above system has the following symmetries:
\begin{eqnarray}
&\fl\label{IIsym1}(\eta^\prime,\eta)\,\rightarrow\,(-\eta^\prime,-\eta)\\
&\fl\label{IIsym2}(\zeta^\prime,\zeta)\,\rightarrow\,(-\zeta^\prime,-\zeta)\\
&\fl\label{IIsym3}(\Sigma_3^\prime,\Sigma_3)\,\rightarrow\,(-\Sigma_3^\prime,-\Sigma_3).\\
&\fl\label{IIsym4}\theta\,\rightarrow\,\theta+\pi
\end{eqnarray}

\subsection{Asymptotic decoupling in the matter content}
\label{subsec:assDec}
For equilibrium points, we must require that all scalars are constants. Henceforth we must generally require $\theta^\prime=0$. This requires $\theta=n\,\pi/2$, where $n\,\in\,\mathbb{Z}$, in which case one finds from the parametrization that $\Theta$ and $V_1$ decouple. This is expected, since $\Theta$ is a monotonically decreasing function for $q<2$. The only way to avoid decoupling would be to enforce $\Sigma_-=\Sigma_-^\prime=0$, while generally $\zeta,\eta\,\neq\,0$. In this case, however, $\zeta$ decreases monotonically. Consequently, the eq. for  $\Sigma_-^\prime$ enforces $\Sigma_3^2=\eta^2=0$ asymptotically. As a result, also $\Sigma_+$ decreases monotonically, and henceforth $\Omega_{\rm pf}\,\rightarrow\,1$ (which is flat FLRW). Otherwise, $\Theta$ and $V_1$ decouple. Furthermore, since $\theta=n\pi$ corresponds to $\mathcal{B}$(I), any new equilibrium sets visible in this gauge will be found in $\theta=n\,\pi/2$.

We also note that the stability of the variables except $\theta$ cannot be sensitive to the number $n$, since only trigonometric squares appear in these equations. For $\theta$, on the other hand, $n$ will be important. There are, however, only two physically distinct options: $\{n\,\in\,$odd $\}$ or $\{n\,\in\,$even $\}$. Mathematically it therefore suffices to study $n=0$ and $n=1$. From Lemma \ref{lemBII2} and from the discussion in the current paragraph, we recognize that the flow along $\theta$ for $\Sigma_-<0$ is important. As displayed in Figure \ref{fig:flowth}, the flow in this case is always towards $\theta=\left(n+1/2\right)\pi$. 
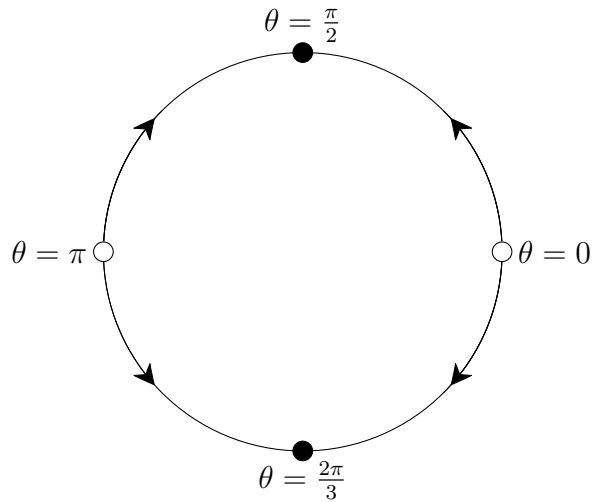
\begin{figure}[H]
\centering
\begin{tikzpicture}[scale=5.3,cap=round,>=latex]
  \def\Radius{0.5cm}

  \draw (0cm,0cm) circle[radius=\Radius];

  \begin{scope}[
    -{Stealth[round, length=8pt, width=8pt, bend]},
    shorten >=4pt,
    very thin,
  ]
    \draw (-45:\Radius) arc(-45:45:\Radius);
    \draw (180:\Radius) arc(180:135:\Radius);
    \draw (45:\Radius) arc(45:-45:\Radius);
    \draw (135:\Radius) arc(135:225:\Radius);
  \end{scope}

\draw[white, fill=white] (180:\Radius) circle(.7pt) node[] {};
\draw[black] (180:\Radius) circle(.7pt) node[left] { $\theta=\pi\,$};
\draw[white, fill=white] (0:\Radius) circle(.7pt) node[] {};
\draw[black] (0:\Radius) circle(.7pt) node[right] { $\,\theta=0$};

\draw[black, fill=black] (90:\Radius) circle(.7pt) node[above] { $\theta=\frac{\pi}{2}$};
\draw[black,fill=black] (270:\Radius) circle(.7pt) node[below] {$\theta=\frac{2\pi}{3}$};
\end{tikzpicture}
\caption{A schematic of the flow along $\theta$ for negative $\Sigma_-$. Black (white) dots represent stable (unstable) subsets of the dynamical system. This flow is generic in the invariant subset $\zeta=0$ for negative $\Sigma_-$.}
\label{fig:flowth}
\end{figure}

\subsection{Equilibrium sets and stability}
Table \ref{tab:IINparam} presents the equilibrium sets found in $\mathcal{B}$(II) in the $N_-$\,-\,gauge with the parametrization \eref{param}. The stability analysis is summarized in Table~\ref{tab:StabII}.

\begin{table}[H]
	\centering
	\resizebox{\textwidth}{!}{\begin{tabular}{llccccccccc}
			\toprule
			\multicolumn{11}{c}{\textbf{Equilibrium sets of $\mathcal{B}$(II) analysed in $N_-$\,-\,gauge.}} \\
			\hline
Set&$\mathcal{P}$& $q$ & $\gamma $&$\Omega_{\rm pf}$&$\Sigma_+$&$\Sigma_-$&$\Sigma_3$&$\eta$&$\zeta$&$\theta$\\
\hline
$\mathcal{S}^0$(I)& flat FLRW& $-1+\frac{3}{2}\gamma$& $[0,2)$&1&0&0&0&0&0&$(0,2\pi]$\\
$\mathcal{C}^0$(I)& JED($\beta_1,\beta_2$) &2& $[0,2]$&0&$\beta_1\in[-1,1]$&$\beta_2\,\in\,[-\sqrt{1-\beta_1},\sqrt{1-\beta_1}]$&0&0&$\zeta\,\in\,[-\sqrt{1-\beta ^2},-\sqrt{1-\beta^2}$]&$n\pi$ \\
$\mathcal{C}^0$(I)& JS($\beta_1,\beta_2,\zeta$)&2& $2$&$[0,1-\zeta^2-\beta^2]$&$\beta_1\in[-1,1]$&$\beta_2\,\in\,[-\sqrt{1-\beta_1},\sqrt{1-\beta_1}]$&0&0&$\zeta\,\in\,[-\sqrt{1-\beta ^2},-\sqrt{1-\beta^2}$]&$n\pi$ \\
$\mathcal{S}^+$(I)&W&$-1+\frac{3}{2}\gamma$&$(\frac{2}{3},2)$&$-\frac{3}{4} (\gamma -2)$&$\frac{1}{2}-\frac{3}{4}\gamma$&0&0&$\pm\frac{3}{4} \sqrt{(2-\gamma ) \left(\gamma -\frac{2}{3}\right)}$&0&$\left(n+\frac{1}{2}\right)\pi$\\
$\mathcal{D}^+$(I)&R&$-1+\frac{3}{2}\gamma$&$(\frac{6}{5},\frac{4}{3})$&$6-\frac{9}{2}\gamma$&$\frac{1}{2}-\frac{3}{4}\gamma$&$\frac{\sqrt{3}}{4}\left(6-5\gamma\right)$&$\pm\frac{1}{2}\sqrt{\frac{15}{2}\left(2-\gamma\right)\left(\gamma-\frac{6}{5}\right)}$&$\pm\frac{1}{2}\sqrt{\frac{27}{2} (2-\gamma) (\gamma -\frac{10}{9})}$&0&$\left(n+\frac{1}{2}\right)\pi$\\
$\mathcal{D}^+$(I)&E&$1$&$(0,2)$&0&$-\frac{1}{2}$&$-\frac{1}{2\sqrt{3}}$&$\pm\,\frac{1}{\sqrt{6}}$&$\pm\frac{1}{\sqrt{2}}$&0&$\left(n+\frac{1}{2}\right)\pi$\\
$\mathcal{C}^0$(II)&CS&$-1+\frac{3}{2}\gamma$&$(\frac{2}{3},2)$&$-\frac{3}{16}\left(\gamma-6\right)$&$\frac{1}{16}(2-3\gamma)$&$\frac{\sqrt{3}}{16}(2-3\gamma)$&0&$\pm\frac{3}{8} \sqrt{(2-\gamma ) \left(\gamma -\frac{2}{3}\right)}$&0&$n\pi$\\
\hline

	\end{tabular}}
	\caption{Summary of equilibrium sets in $\mathcal{B}$(II). Here $\beta^2\,\equiv\,\beta_1^2+\beta_2^2$, and $n$ is an integer.}
	\label{tab:IINparam}
\end{table}
\noindent Table \ref{tab:StabISS1} gives the eigenvalues of the linearized matrix around each equilibrium set given in \ref{tab:IINparam}. The stability of some equilibrium sets (marked with "not param." in the table) were best studied without the parametrization (but still in $N_-$ gauge. In particular this is true for FLRW and Kasner, since they belong to the origo of the parametrization.

\paragraph{FLRW:} This equilibrium set is in the origo of the parametrization. For this particular point, we therefore used a different gauge and parametrization to confirm our results from N$_-$\,-\,gauge. To prevent the paper from becoming any longer, we do not report further on this issue: the results agree. The particular point $\gamma=2/3$ required the use of center-manifold analysis, which was conducted. The result is that FLRW is stable for $\gamma\,\leq\,2/3$, which is in excellent agreement with Theorem \ref{genThmBII}.

\paragraph{Jacobs' Extended Disk, JED($\beta_1,\beta_2$), and Jacobs' Sphere, JS($\beta_1,\beta_2,\Theta$).} These two equilibrium sets have the same eigenvalues, except for the $\gamma$-dependent one, which is 0 for JS, since $\gamma=2$. The stability regions fall into the same categories for these Eq. sets, however. In particular they are repellers for $\beta_1,\beta_2<0\,\cap\,\frac{\abs{\beta_2}}{\sqrt{3}}<\abs{\beta_1}<1-\sqrt{3}\abs{\beta_2}$.

\paragraph{Collins-Stewart, CS:} This equilibrium set is stable w.r.t. perturbations in all variables except $\theta.$ It is therefore a saddle. In the invariant subset $V_1=\Sigma_\times=0$ (corresponding to $\theta=n\pi$) it is an attractor, however. This is in agreement with earlier works on $\mathcal{B}$(II) with a perfect fluid.

\paragraph{Wonderland (W):} This is the $\mathcal{B}$(I) version of Wonderland. The symmetry \eref{IIsym1} of the dynamical system ensures that the stability is not affected by the sign of $\eta$. The Wonderland equilibrium set has one eigenvalue that cannot be accounted for by free parameters. The center manifold is locally in the $\theta$ direction, and $\Sigma_-=0$. Hence we must perform a center-manifold analysis to understand the flow in this direction. Doing so, we find to leading order $$\theta^\prime=-\frac{9}{2}(\gamma-\frac{2}{3})\left(\theta\,\pm\frac{\pi}{2}\right)^3.$$ With the equilibrium sets at $\theta=\pm\frac{\pi}{2}$, we therefore find that both copies of Wonderland are stable. 
\paragraph{Rope: } The Rope also has several copies. For $\theta\,\in\,[0,2\pi)$ there are 8 copies, since one or more of the changes in the set $\{\theta\,\rightarrow\theta+\pi,\Sigma_3\,\rightarrow-\Sigma_3,\eta\,\rightarrow\,-\eta\}$ will produce a different copy of the equilibrium set Rope.  The symmetries \eref{IIsym2},\eref{IIsym3} and \eref{IIsym4} ensure that the stability is the same. The eigenvalues reveal that the Rope is an attractor on all of its existence. 

\paragraph{Edge: }Also the Edge has 8 different copies for $\theta\,\in\,[0,2\pi)$, since one or more of the changes in the set $\{\theta\,\rightarrow\theta+\pi,\Sigma_3\,\rightarrow-\Sigma_3,\eta\,\rightarrow\,-\eta\}$ will produce a different copy of the equilibrium point  Edge. With the same argument as for Rope, it therefore suffices to study one copy. Since there is one zero-eigenvalue not accounted for by free parameters, center-manifold analysis is required. The center manifold is one-dimensional and in the $\zeta$ -direction. In doing the analysis, one finds that the flow on the center manifold is to leading order governed by
$$\zeta^\prime=-12\zeta^3.$$
Hence we may conclude from the remaining eigenvalues that the Edge is an attractor for $\gamma\,>\,4/3$ and a saddle point for $\gamma<4/3$. The particular point $\gamma\,=\,4/3$ remains uncertain.

\begin{table}[H]
\centering
\resizebox{\textwidth}{!}{\begin{tabular}{lcclclclclc}
\toprule
\multicolumn{11}{c}{\textbf{Classification of equilibrium sets in $\mathcal{B}$(II)}} \\
\hline
Eq.set&& Existence && Attractor && Saddle && Repeller && Inconclusive \\
\hline
flat FLRW && $\gamma\in[0,2)$ &&$\gamma\,\in\,[0,\frac{2}{3}]$&&$\gamma\,\in\,(\frac{2}{3},2)$&&&&$\gamma=\frac{2}{3}$\\
JED($\beta_1,\beta_2$) &&  $\gamma\,\in\,[0,2)$ &&&&else&&$\beta_1,\beta_2<0\,\cap\,\frac{\abs{\beta_2}}{\sqrt{3}}<\abs{\beta_1}<1-\sqrt{3}\abs{\beta_2}$&&$\beta_2=0\,\cup\,\beta_2=\sqrt{3}\beta_1\,\cup\,\beta_1=-\sqrt{3}\beta_2-1$\\
JS($\beta_1,\beta_2,\zeta$)&& $\gamma=2$ &&&&else&&$\beta_1,\beta_2<0\,\cap\,\frac{\abs{\beta_2}}{\sqrt{3}}<\abs{\beta_1}<1-\sqrt{3}\abs{\beta_2}$&&$\beta_2=0\,\cup\,\beta_2=\sqrt{3}\beta_1\,\cup\,\beta_1=-\sqrt{3}\beta_2-1$\\
CS&& $\gamma\in(\frac{2}{3},2)$ &&&&$\forall\,\gamma$ &&&&  \\
W&&$\gamma\in(\frac{2}{3},2)$&&$\gamma\,\in\,(\frac{2}{3},\frac{6}{5})$&&&&&&$\gamma=\frac{6}{5}$ \\
R&&  $\gamma\in(\frac{6}{5},\frac{4}{3})$&&$\forall\gamma$&&&&&& \\
E&&  $\gamma\in(0,2)$&&$\gamma\,\in\,(\frac{4}{3},2)$&&&&&&$\gamma=\frac{4}{3}$ \\\hline
	\end{tabular}}
	\caption{The domains where the stability analysis is conclusive are divided into attractor, saddle and repeller subdomains by the conditions above.  The rightmost column shows the domains where the linear stability analysis is inconclusive.}
	\label{tab:StabII}
\end{table} 

\paragraph{Global results.} We summarize and extend the results of the analysis of $\mathcal{B}$(II) in the following two lemmas.
\begin{lem}[Future asymptotes of $\mathcal{B}$(II)]
\label{lemBII3}
The set $\mathcal{B}$(II) with a perfect fluid and a $j$-form fluid with $V_1^2>0$ is future asymptotic to $\mathcal{B}$(I) for $0\,\leq\,\gamma\,\leq\,2$.
\end{lem}
\begin{proof} Theorems \ref{NoHair} and \ref{genThmBII} cover the intervall $0\,\leq\,\gamma\,\leq\,2/3$. Lemma \ref{lemBII1} and \ref{lemBII2} establish that $\mathcal{B}$(II) is asymptotic to $\mathcal{B}$(II)$\vert_{N_-=0}\,\cup\,\mathcal{B}$(II)$\vert_{\Sigma_\times=0,\,\Sigma_-\,\leq\,0}$. The first of these sets is $\mathcal{B}$(I). Furthermore, we have shown in Sec. \ref{subsec:assDec} that $\Sigma_-^\prime=\Sigma_-=0$ must correspond asymptotically to $\Omega_{\rm pf}=1\,\subset\mathcal{B}$(I) for $\theta\,\notin\,\{n\pi\}$. Otherwise, for $\Sigma_-<0$, the same discussion showed that for $\theta\,\notin\,\{n\pi\}$ (where $n\,\in\,\mathbb{Z}$), the flow will be towards  $\{\theta=(n+1/2)\pi\}$, which is again $\mathcal{B}$(I). Therefore $\{\theta\,=\,n\pi\}$, which is Collins-Stewart (CS), is the only exception, and here $V_1^2\,=\,0$. The Lemma follows.
\end{proof}
Note that the case where $V_1=0$ is asymptotically equivalent to the case with a perfect fluid only, since $\Theta$ decays monotonically. That case has already been investigated (cf. Prop. 6.1 in \cite{dynSys}), and CS is found to be the global future attractor for $\Omega_{\rm pf}\,>0$ and $2/3\,<\,\gamma\,<\,2$. This is in perfect agreement with our local analysis. 

\begin{lem}[Past asymptoticity of $\mathcal{B}$(II)]
The set $\mathcal{B}$(II) with a $j$-form fluid and a perfect fluid with $\Omega_{\rm pf}>0$ is for $\gamma<2$ ($\gamma=2$) and $\Theta^2\,>\,0$ past asymptotic to JED (JS).
\end{lem}
\begin{proof}
The available past attractors in conjunction with the monotonic function $\Theta^\prime=(q-2)\Theta$.
\end{proof}
The case $\Theta=V_1=0$ is covered by Proposition 6.2 in \cite{dynSys}. Also, the general behavior $\mathcal{B}$(I)  with a perfect fluid and a $j$-form fluid is already studied in our previous work. Hence Theorem \ref{anisBII} follows from the two lemmas above.

\paragraph{Heteroclinic sequence.} 
CS is the global future attractor for $V_1=0$  and JED is the global past attractor. The heteroclinic sequences established withouth the $j$-form in \cite{dynSys} (Figs. 6.8 and 6.9), should carry over. In particular, based on the discussion above, there must exist a heteroclinic sequence
\begin{equation}
\rm JED\,\longrightarrow F\,\longrightarrow\,CS\,\longrightarrow\,W.
\end{equation}

       \begin{figure}[h]
     \subfloat[$\mathcal{B}$(II). \label{subfig:B2}]{%
       \includegraphics[width=0.5\textwidth]{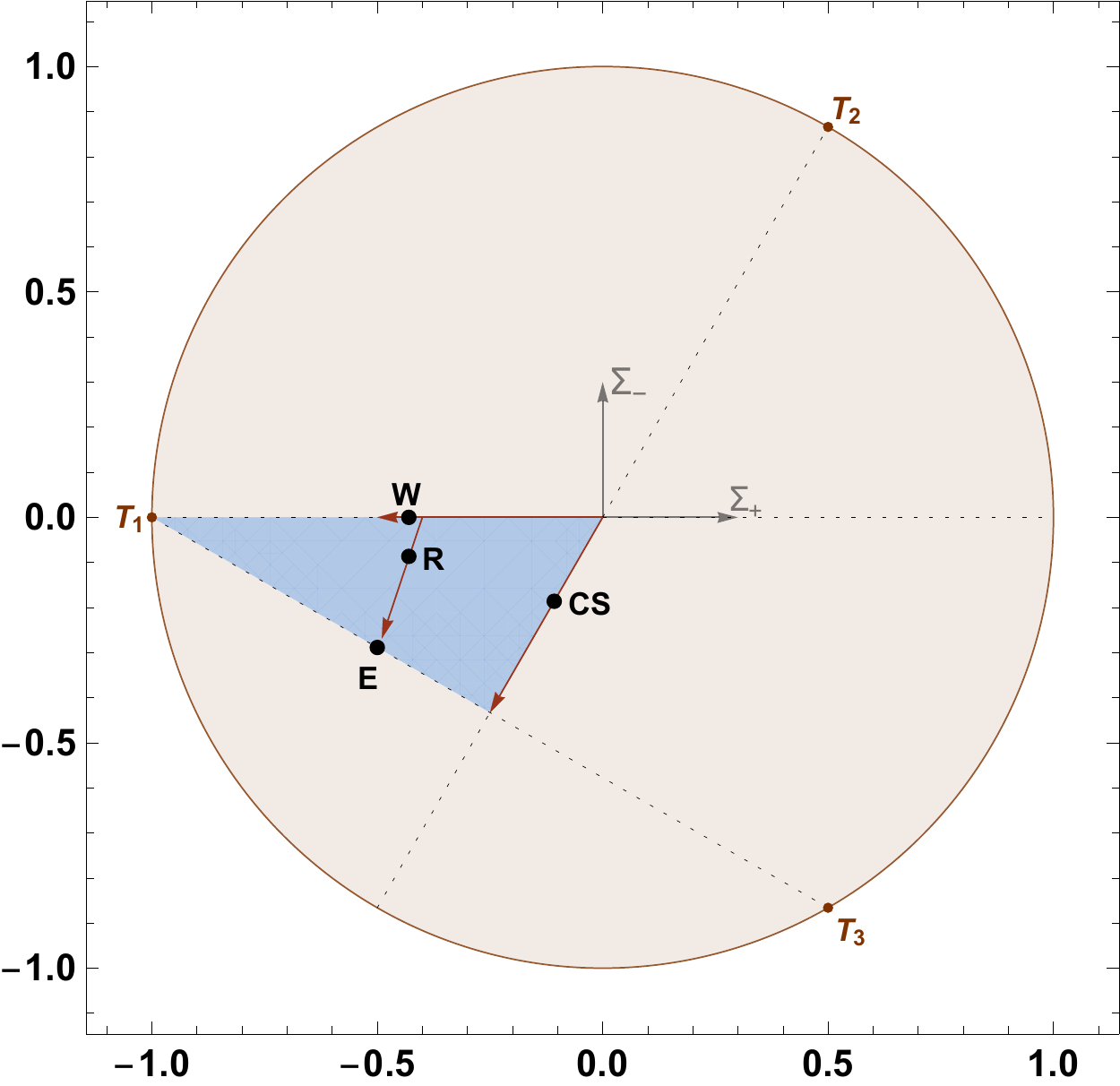}
     }
     \hfill
     \subfloat[$\mathcal{B}$(IV).\label{subfig:B4}]{%
       \includegraphics[width=0.5\textwidth]{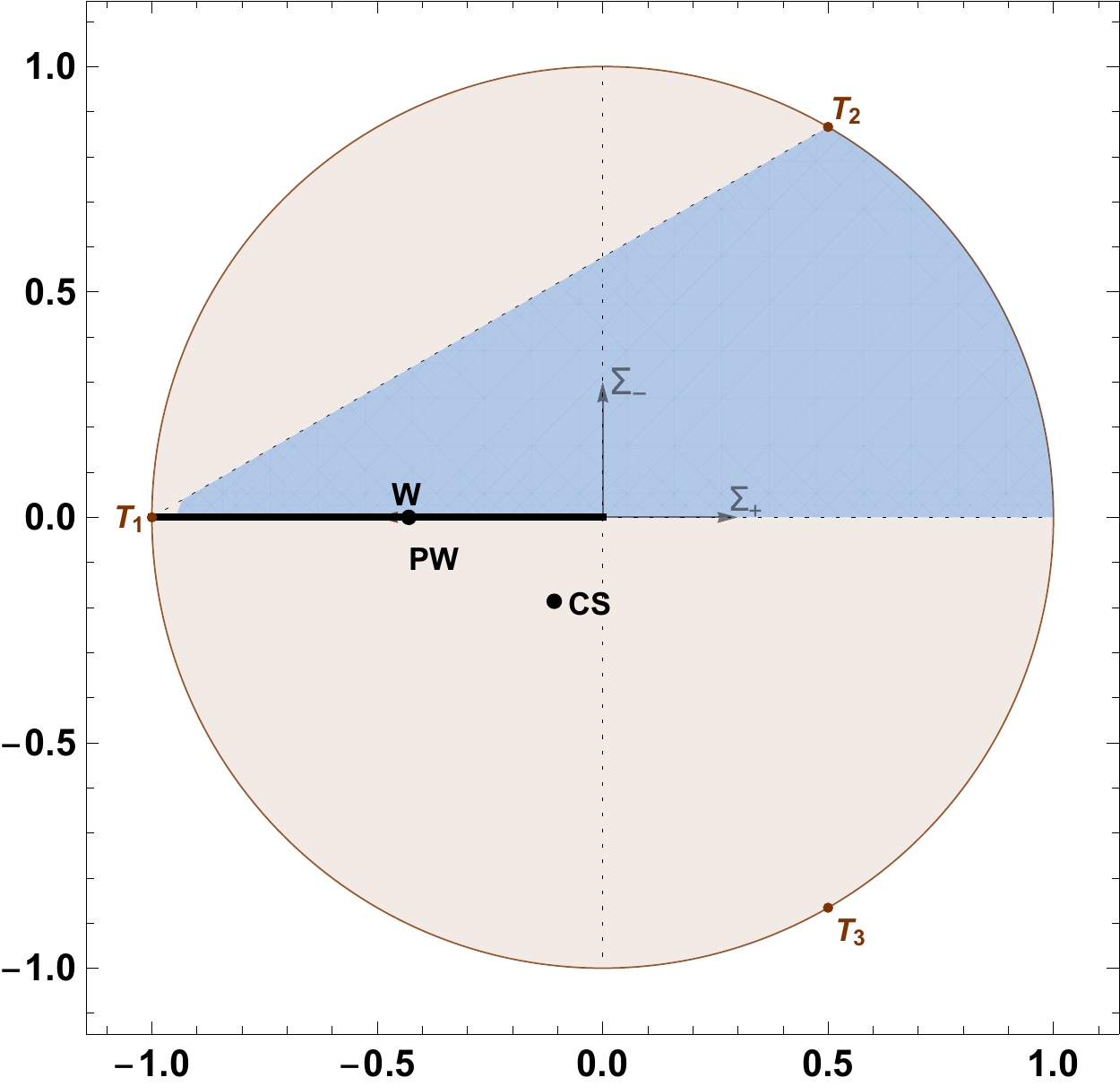}
     }
     \caption{The $(\Sigma_+,\Sigma_-)$ -plane, showing the shadow of the JED disk ($\beta_1^2+\beta_2^2<1$) and the Kasner circle ($\Sigma_+^2+\Sigma_-^2=1$). The blue-shaded regions inside the circles represent the projection of the stable part of JED. The dashed lines show where the first order stability analysis is inconclusive. $T_i$ represent Taub point $i$. The equilibrium points (black) are shown for $\gamma=1.24$, on top of arrows giving the line along which each equilibrium point shift its position as we very gamma from minimum to maximum for each eq.p.}
     \label{fig:B24Kasner}
   \end{figure}
   
\section{Conclusions} 
\label{Sec:Concl}
We have studied in detail the Bianchi invariant sets $\mathcal{B}$(II), $\mathcal{B}$(IV), $\mathcal{B}$(VII$_0$) and  $\mathcal{B}$(VII$_h$) with a perfect fluid and a $j$-form fluid (where $j\,\in\,\{1,3\}$). We have localized exact solutions to Einstein's field equations that represent stable, anisotropic space-times in all these sets for $2/3<\gamma<2$. In particular, Wonderland exists in all these sets and is stable on all of its existence, except in $\mathcal{B}$(II), where Edge and Rope takes over the stability for some parameter values. Moreover, in $\mathcal{B}$(IV) and $\mathcal{B}$(VII$_h$) plane gravitational wave solutions (with a non-zero $p$-form) serve as attractors whenever  $2/3<\,\gamma\,<2$. The results are summarized in Section \ref{Sec:Summary}. Anisotropic attractors in $\mathcal{B}$(VI) is left for future work.

\section{Acknowledgments}
 We take the opportunity to thank M. Thorsrud for useful discussions.
This work was supported through the Research Council of Norway, Toppforsk
grant no. 250367: \emph{Pseudo-Riemannian Geometry and Polynomial Curvature Invariants:
Classification, Characterisation and Applications.}

\newpage\phantom{0}
\newpage
\appendix
\section{Complex variables}
\label{App:CompVar}
In this appendix, we briefly recapitulate how the expansion-normalized variables used in this paper are formed. Refer to \cite{normann18} for further explanations. The expansion-normalization is as follows. 
\begin{eqnarray}
\label{defENC}
&\fl\Sigma_{+}=\frac{\sigma_+}{H}\,,\quad\quad \Omega_i=\frac{\rho_i}{3H^2}\phantom{0},\quad\quad A_i=\frac{a_i}{H}\phantom{000.},\nonumber\\
&\fl\Sigma_{-}=\frac{\sigma_-}{H}\,,\quad\quad \Omega_\Lambda=\frac{\Lambda}{3H^2}\,,\quad\quad N_{+}=\frac{n_+}{H}\phantom{00.},\nonumber\\
&\fl\Sigma_{\times}=\frac{\sigma_\times}{H}\,,\quad\quad V_i=\frac{v_i}{\sqrt{6}H}\,,\quad\quad N_{-}=\frac{n_-}{H}\phantom{00.},\\
&\fl\Sigma_2\,=\,\frac{\sigma_{2}}{H}\,,\quad\quad\Theta=\frac{w}{\sqrt{6} H}\,,\quad\quad N_{\times}=\frac{n_\times}{H}\phantom{00.},\nonumber\\
&\fl\Sigma_3\,=\,\frac{\sigma_{3}}{H}\,,\quad\quad\Xi_i=\frac{q_i}{3H^2}\,\phantom{0},\quad\quad\Sigma^2=\frac{\sigma_{ab}\sigma^{ab}}{6 H^2}.\nonumber
\end{eqnarray}
where $H$ is the Hubble parameter.
In this way the equations of motion become an autonomous system of differential equations and all equilibrium points will represent self-similar cosmologies. The above definitions differ slightly from other authors (e.g. \cite{coley05, dynSys}), since we decompose such that
\label{app:Decomp}
\begin{equation}
x_{ab}=
\left(\begin{array}{ccc}
-2x_+      & \sqrt{3}x_2&\sqrt{3}x_3\\
\sqrt{3}x_2     & x_++\sqrt{3}x_-&\sqrt{3}x_\times\\
\sqrt{3}x_3     & \sqrt{3}x_\times&x_+-\sqrt{3}x_-\\
\end{array}\right),
\end{equation}
where $x_{ab}$ is one of the trace-less matrices $n_{ab}$ or $\sigma_{ab}$ (their normalized equivalents $N_{ab}$ and $\Sigma_{ab}$ have the same structure). Note that $n_{1b}=0$ (for all $b$) for the considered Bianchi type I-VII$_h$.

We align our frame such that the basis vectors $\mathbf{e}_A$ (where $A=\{2,3\}$) are aligned with the orbits of the $G_{2}$ subgroup permitted by the isometry group in the Solvable Bianchi types. This 1+1+2 split of space-time effectually fixes the two (expansion-normalized) rotations $R_2$ and $R_3$, but leaves a rotational gauge freedom $R_1$. More specifically, $R_1$ is the rotation of the frame around the $\mathbf{e}_1$-axis), which is orthogonal to the orbits of the $G_{2}$ subgroup. Taking the angle $\phi$ to be constant on the orbits of $G_{2}$ the (expansion-normalized) local angular velocity $R_{1}$ of a Fermi-propagated axis with respect to the triad $\bf{e}_{a}$ is given as
\begin{equation}
R_1\,=\phi^\prime.
\end{equation}
Following~\cite{coley05}, we leave this gauge-freedom in the equations. Finally note the definitions
\begin{eqnarray}
\label{complexVar}
\eqalign{
\mathbf{N}_\Delta=N_-+iN_\times\,,\quad\quad &\mathbf{V}_c=V_2+iV_3\,,\\
\mathbf{\Sigma}_\Delta=\Sigma_-+i\Sigma_\times\,,\quad\quad &\mathbf{\Sigma}_1=\Sigma_2+i\Sigma_3\,.}
\end{eqnarray}

\newpage
\section{Comment on eigenvalue problem for Wonderland}
\label{App:WonderlandStab}
We shall note the following very useful observation. In computing the local stability around each equilibrium set $P$, we find from\eref{FluidEqs}-\eref{JacId2} a system on the form
\begin{equation}
X^\prime={\rm D}f(X_0)\,X\phantom{000}\rightarrow\phantom{000}\det({\rm D}f(X_0)-l\,I)=0.
\end{equation} 
where ${\rm D}f(X_0)\,\in\,\mathcal{M}_{8\times 8}$ (after reduction through constraints) is the jacobian around the Equilibrium set $X_0$. Furthermore let $M_1$ represent the Jacobian of the subsystem for the vector $(\mathbf{N}_\Delta,\mathbf{\Sigma}_\Delta)^{\rm T}.$ For an equilibrium point where $\mathbf{V}_c=\mathbf{\Sigma}_1=0$, the matrix $M_1$ reads
\begin{equation}
M_1=\left[ {\begin{array}{cc}
q+2\Sigma_+-2iR_1&2 N_+\\
    -2\left(iA+N_+\right) & q-2-2iR_1
\end{array} } \right],
\end{equation}
One may now find from the set of equations \eref{FluidEqs}-\eref{JacId2} that if $(\mathbf{N}_\Delta,\mathbf{\Sigma}_\Delta)\rightarrow 0$, the full linearization ${\rm D}f(X_0)$ around the equilibrium point $X_0$ is on block diagonal form ${\rm D}f(X_0)\,\in\,\mathcal{M}_{4\times 4}\times\mathcal{M}_{4\times 4}$ with $M_1$ as one of the blocks. Specifically, from determinant rules for block matrices we may use that 
\begin{equation}
\label{block}
\det({\rm D}f(X_0)-l\,I)=\det(M_1-l\,I)\cdot\det(M_2-l\,I)=0,
\end{equation}
Here $M_2$ represents the 4 by 4 matrix for the rest of the system. Hence we may compute the local stability for $M_1$ and $M_2$ separately. We may use this for the equilibrium set called Wonderland, as shown in the following.

\subsection*{Wonderland in $\mathcal{B}$(VII$_h$)}
(Some of) the specifications for Wonderland in $F$\,-\,gauge ($R_1=0$) are
\begin{equation}\fl
\label{WSpecss}
q=-1+\frac{3}{2}\gamma,\phantom{000}\Sigma_+=\frac{1}{2}-\frac{3}{4}\gamma,\phantom{000}A=\frac{3}{4}(\gamma-2)\kappa,\phantom{000}N_+=\nu_1.
\end{equation}
The parameter $\kappa$ is restricted such that $-1\,<\kappa\,\,\leq\,0$. The linearisation matrix $M_1$ now takes the form
\begin{equation}
M_1=\left[ {\begin{array}{cc}
0&2 N_+\\
    -2\left(iA+N_+\right) & \frac{3}{2}\left(\gamma-2\right)
\end{array} } \right],
\end{equation}
where we delay inserting for $N_+$ and $A$ for practical purposes only. Solving the characteristic equation we find the eigenvalues
\begin{equation}
l_{\pm}=-\frac{3}{4}\left(2-\gamma\right)\,\pm\,\sqrt{c}.
\end{equation}
where $c=a+ib$ and
\begin{eqnarray}
a=\left(\frac{3}{4}(\gamma -2)\right)^2-4N_+^2\quad\quad\quad\textrm{and}\quad\quad\quad b=-4A N_+.
\end{eqnarray}
We see that $l_+$ has always negative real part, but $l_-$ is a bit more elusive. To find out we go to polar coordinates in order to express the real part $M$ of the complex number $\sqrt{c}$ as
\begin{equation}
M=\sqrt{\abs{c}}\cos\left(\frac{\phi}{2}\right)=\sqrt{\frac{\abs{c}+a}{2}},
\end{equation}
where the last equality follows from a simple geometrical argument, and where $\abs{c}=\sqrt{cc^*}=\sqrt{a^2+b^2}$. One may readily show that 
\begin{equation}\fl
M=2^{-1/2}\sqrt{\sqrt{9 (\gamma -2)^2 \kappa ^2 \nu _1^2+\left(\frac{9}{16} (\gamma -2)^2-4 \nu _1^2\right){}^2}+\frac{9}{16} (\gamma -2)^2-4 \nu_1^2}.
\end{equation}
This is a purely real number, and one finds that 
\begin{equation}
Re\{l_{-}\}=-\frac{3}{4}\left(2-\gamma\right)\,\pm\,M<0\quad\forall\quad \gamma,\kappa\quad\textrm{and}\,\nu_1\,\neq\,0.
\end{equation}
Since $\abs{c}$ and $a$ are invariant under the complex conjugation $l_{\pm}\,\rightarrow\,l_{\pm}^*$, the condition for the vectors ($\mathbf{N}_\Delta^*,\mathbf{\Sigma}_\Delta^*$) to be stable will necessarily be the same as the condition for ($\mathbf{N}_\Delta,\mathbf{\Sigma}_\Delta$) to be stable. The  eigenvalues
\begin{equation}
\{l_+,l_-,l_+^*,l_-^*\}
\end{equation}
therefore all have negative real parts in the whole parameter domain except if $\nu_1=0$. The particular value $\nu_1=0$ corresponds to Wonderland in $\mathcal{B}$(V), W($\kappa$). If $\nu_1=\kappa=0$, then it is Wonderland in $\mathcal{B}$(I), W. In this limiting case, $l_-,l_-^*$ will both be zero. This is because $\nu_1=0$ is the point where W$(\kappa,\nu_1)$ meets W$(\kappa,\nu^2$); Wonderland in $\mathcal{B}$(VI$_h$).

\subsubsection*{The remaining matrix M$_2$}
Having inserted the constraints, the remaining physical part of state space is four dimensional. $M_2$ has one zero-eigenvalue, which can be shown to be in the $N_+$-direction. This zero is accounted for by the group parameter $h$: Every $h$ corresponds to an invariant set, so $h^\prime=0$. But $h$ is nevertheless a parameter of the dynamical system and hence of the equilibrium sets. It will therefore result in one extra zero-eigenvalue. The remaining part of the matrix $M_2$ (when $(\mathbf{N}_\Delta,\mathbf{\Sigma}_\Delta)\rightarrow 0$), is $\mathcal{B}$(V) modulo $\Sigma_-$. Now, $\Sigma_-=0$ corresponds to an invariant subspace in $\mathcal{B}$(V). We may therefore use the results of our study of $\mathcal{B}$(V) in \cite{normann18}, where the eigenvalues were found to be all negative except for one, which was zero, corresponding to the parameter $\kappa$.

The specification $A=\Theta=0$ in the Wonderland specifications (corresponding to $\kappa=0$), is also a subspace of $\mathcal{B}$(V), and again we can use the results of the $\mathcal{B}$(V) analysis. 

\section{Monotonic functions}
\label{App:B}
The following monotonic functions are known to exist.

\begin{itemize}
\item In the sets $\mathcal{B}$(VI$_0$) and $\mathcal{B}$(VII$_0$) models,
\begin{eqnarray}
\label{b1}Z_6=\frac{V_1^{3\gamma-2}\Omega^2}{\phi^{3\gamma+2}}, \qquad \phi=1+m\Sigma_+, \qquad m=\frac 14(3\gamma-2), \\
\label{b2}\frac{Z_6^\prime}{Z_6}=\phi^{-1}\left[8(\Sigma_++m)^2+\frac 32(3\gamma+2)(2-\gamma)\left(|{\mathbf\Sigma}_{\Delta}|^2+\Theta^2\right)\right]
\end{eqnarray} 
\item In the set $\mathcal{B}$(II), we find that $\Theta$ decreases monotonically, as discussed in the text.Additionally, we find the following.
\begin{itemize}
\item The monotonic function $Z_1$ which increases for $\gamma\,>\,4/3$, and decreases for  $\gamma\,<\,4/3$:
\begin{eqnarray}
Z_1&=&\frac{\Sigma_\times\Sigma_3^2V_1^3}{\Omega^3_{\rm pf}}, \qquad Z_1^\prime=3(3\gamma-4)Z_1,
\label{z1t} 
\end{eqnarray}
\item The monotonic function $Z_2$ which increases for $\gamma\,<\,2$:
\begin{eqnarray}
\fl\label{z2a} Z_2&=&\frac{\Omega_{\rm pf}^{1/\gamma} V_1}{\Sigma_\times N_-},\\
\fl \label{z2b}Z_2^\prime&=&\frac{2-\gamma}{2\gamma}\left(2(1-\Omega_{\rm pf})+3\gamma\Omega_{\rm pf}+4 \left(\Theta ^2+ \Sigma _-^2+ \Sigma _+^2+\Sigma _\times^2\right)\right) Z_2.
\end{eqnarray}
\item The monotonic function $Z_3$, which decreases for $\gamma\,>\,6/5$ and increases for $\gamma<6/5$:
\begin{eqnarray}
\fl \label{z3a} Z_3&=&\frac{\Omega_{\rm pf}^{3/\gamma}N_-}{\Sigma_3^2V_1^4},\\
\fl \label{z3b}Z_3^\prime&=&\frac{6-5\gamma}{2\gamma}\left(2(1-\Omega_{\rm pf})+3\gamma\Omega_{\rm pf}+4 (\Theta ^2+\Sigma _-^2+ \Sigma _+^2)\right) Z_3. 
\end{eqnarray}
\end{itemize}
\end{itemize}

\section{Eigenvalues}
\label{App:Eigenvalues}
The zero-eigenvalues correspond either to parameters of the equilibrium set, or they result from an inconclusive first order stability analysis. In the tables to follow in this appendix, the rightmost column ($\#$) gives the number of 0-eigenvalues that can be accounted for by parameters of the Eq. set. The column marked `s' indicates whether the analysis has been performed in the extended (e) state space, or in the physical (p) state space.
\subsection*{Eigenvalues of equilibrium sets in $\mathcal{B}$(VII$_h$)}
\begin{table}[H]
	\centering
	\resizebox{\textwidth}{!}{\begin{tabular}{llllr}
			\toprule
			\multicolumn{5}{c}{\textbf{Eigenvalues of equilibrium sets in $\mathcal{B}$(VII$_h$) in $N_-$\,-\,gauge.}} \\
			\hline
$\mathcal{P}$&$\gamma$ &s&eigenvalues $\{l_1,l_2,l_3,l_4,l_5,l_6,l_7,l_8\}$&$\#$\\
\hline
CS & $(\frac{2}{3},2)$&p&$\left\{-\frac{3}{2}(2-\gamma),-\frac{3}{2}(2-\gamma),\frac{9}{8} ( \gamma -\frac{2}{3}),\frac{9}{8} ( \gamma -\frac{2}{3}),\frac{9}{4} ( \gamma -\frac{2}{3}),-\frac{3}{4} \left(2-\gamma \,\pm\,\sqrt{-6 \gamma ^3+56 \gamma ^2-120 \gamma +64}\right)\right\}$&0\\
PW($\alpha,\beta_1,\nu^2)$&  $[0,2]$&p&$\left\{0,0,0,-2 \left(\beta _1+1\right),-4 \beta _1-3 \gamma +2,-2(\beta_1+1)\,\pm\,i\,4\sqrt{3}\alpha\nu\right\}$&3\\
   M& $[0,2]$&p&$\{0,0,0,-2,-2,-2,-3\left(\gamma-\frac{2}{3}\right)\}$&3\\
\\ \hline 
	\end{tabular}}
	\caption{Table of eigenvalues of equilibrium sets $\mathcal{P}$ in $\mathcal{B}$(VII$_h$) calculated in $N_-$\,-\,gauge. Refer to appendix text for explanation of columns.}
	\label{tab:EigVIIhNm}
\end{table}

\begin{table}[H]
	\centering
	\resizebox{\textwidth}{!}{\begin{tabular}{llllr}
			\toprule
			\multicolumn{5}{c}{\textbf{Eigenvalues of equilibrium sets in $\mathcal{B}$(VII$_h$) in $F$\,-\,gauge}} \\
			\hline
$\mathcal{P}$&$\gamma$&s&eigenvalues $\{l_1,l_2,l_3,l_4,l_5,l_6,l_7,l_8,l_9\}$&$\#$\\
\hline
flat FLRW&$[0,2)$ & e&$\left\{\frac{3}{2}(\gamma -2),\frac{3}{2}(\gamma -2),\frac{3}{2}(\gamma -2),\frac{3}{2}(\gamma -2),\frac{3}{2}(\gamma-\frac{2}{3}),\frac{3}{2}(\gamma-\frac{2}{3}),\frac{3}{2}(\gamma-\frac{2}{3}),\frac{3}{2}(\gamma-\frac{2}{3}),\frac{3}{2}(\gamma-\frac{2}{3})\right\}$&\\
open FLRW &$\frac{2}{3}$&p&$\left\{-2,-2,-2,0,0,0,0,0\right\}$&1\\
K$(\beta_1,\beta_2)$&  $[0,2)$ & p&$\left\{0,0,0,2 \left(\beta _1+1\right),2 \left(\beta _1+1\right),-3 (\gamma -2),2 \left(\beta _1-\sqrt{3} \sqrt{1-\beta _1^2}+1\right),2
   \left(\beta _1+\sqrt{3} \sqrt{1-\beta _1^2}+1\right)\right\}$&3\\
JED$(\beta_1,\beta_2,\beta_3)$&$[0,2)$&p&$\left\{0,0,0,-3 (\gamma -2),2 \left(\beta _1+1\right),2 \left(\beta _1+1\right),-2 \left(-\beta _1+\sqrt{3} \sqrt{\beta _2^2+\beta
   _3^2}-1\right),2 \left(\beta _1+\sqrt{3} \sqrt{\beta _2^2+\beta _3^2}+1\right)\right\}$&3\\
JS$(\beta_1,\beta_2,\beta_3,\Theta)$&$2$&p&$\left\{0,0,0,0,2 \left(\beta _1+1\right),2 \left(\beta _1+1\right),2 \left(\beta _1-\sqrt{3} \sqrt{\beta _2^2+\beta _3^2}+1\right),2 \left(\beta
   _1+\sqrt{3} \sqrt{\beta _2^2+\beta _3^2}+1\right)\right\}$&4\\
W($\kappa,\nu_1$)&$\left(\frac{2}{3},2\right)$&e&$
\left\{0,0,l_+,l_-,l_+^*,l_-^*,-\frac{3}{4} \left(2-\gamma\,\pm\,\sqrt{(\gamma -2)^2 \left(6 \gamma  \left(\kappa ^2-1\right)-4 \kappa ^2+5\right)}\right)\right\}$&2\\
\\ \hline 
	\end{tabular}}
\caption{Table of eigenvalues of equilibrium sets in $\mathcal{B}$(VII$_h$) calculated in $F$\,-\,gauge. Refer to \ref{App:WonderlandStab} for details on the Wonderland eigenvalues.  Refer to appendix text for explanation of columns.}
	\label{tab:EigVIIhGg}
\end{table}

\subsection*{Eigenvalues of $\mathcal{B}$(VII$_0$) equilibrium sets}
\begin{table}[h]
	\centering
	\resizebox{\textwidth}{!}{\begin{tabular}{llllr}
			\toprule
			\multicolumn{5}{c}{\textbf{Eigenvalues of equilibrium sets in $\mathcal{B}$(VII$_0$) in $F$\,-\,gauge.}} \\
			\hline
$\mathcal{P}$&$\gamma$ &s&eigenvalues $\{l_1,l_2,l_3,l_4,l_5,l_6,l_7 (,l_8)\}$&$\#$\\
\hline
flat FLRW &  $[0,2]$&e&$\left\{-\frac{3}{2}(2-\gamma),-\frac{3}{2}(2-\gamma),-\frac{3}{2}(2-\gamma),-\frac{3}{2}(2-\gamma),-\frac{3}{2}\left(\frac{2}{3}-\gamma\right),-\frac{3}{2}\left(\frac{2}{3}-\gamma\right),-\frac{3}{2}\left(\frac{2}{3}-\gamma\right),-\frac{3}{2}\left(\frac{2}{3}-\gamma\right)\right\}$&0\\
JED($\beta_1,\beta_2,\beta_3$) &$[0,2]$&p&$\left\{0,0,0,6-3 \gamma ,2 \left(\beta _1+1\right),2 \beta _1-2 \sqrt{3} \sqrt{\beta _2^2+\beta _3^2}+2,2 \left(\beta _1+\sqrt{3} \sqrt{\beta
   _2^2+\beta _3^2}+1\right)\right\}$&3\\
   JS($\beta_1,\beta_2,\beta_3,\Theta$) &2&p&$\left\{0,0,0,0,2 \left(\beta _1+1\right),2 \left(\beta _1-\sqrt{3} \sqrt{\beta _2^2+\beta _3^2}+1\right),2 \left(\beta _1+\sqrt{3} \sqrt{\beta
   _2^2+\beta _3^2}+1\right)\right\}$&4\\
W($\nu_1$)& $\left(\frac{2}{3},2\right)$&p&$\left\{0,-\frac{3}{4} \left(-\gamma +\sqrt{-(\gamma -2)^2 (6 \gamma -5)}+2\right),\frac{3}{4} \left(\gamma +\sqrt{-(\gamma -2)^2 (6 \gamma
   -5)}-2\right),-\frac{3}{4} \left(\sqrt{(\gamma -2)^2-4 \kappa ^2}-\gamma +2\right),-\frac{3}{4} \left(\sqrt{(\gamma -2)^2-4 \kappa ^2}-\gamma
   +2\right),\frac{3}{4} \left(\sqrt{(\gamma -2)^2-4 \kappa ^2}+\gamma -2\right),\frac{3}{4} \left(\sqrt{(\gamma -2)^2-4 \kappa ^2}+\gamma
   -2\right)\right\}\}$&2\\
CS& $\left(\frac{2}{3},2\right)$&p&$\left\{0,\frac{3 (\gamma -2)}{2},\frac{3 (\gamma -2)}{2},\frac{3}{8} (3 \gamma -2),\frac{3}{4} (3 \gamma -2),-\frac{3 \left(2 (8-3 \gamma ) \gamma
   +\sqrt{2} \sqrt{-(2-3 \gamma )^2 (\gamma -2) (\gamma  (3 \gamma -22)+16)}-8\right)}{8 (3 \gamma -2)},\frac{3 \left(2 \gamma  (3 \gamma-8)+\sqrt{2} \sqrt{-(2-3 \gamma )^2 (\gamma -2) (\gamma  (3 \gamma -22)+16)}+8\right)}{8 (3 \gamma -2)}\right\}$&0\\
\\ \hline 
	\end{tabular}}
	\caption{Table of eigenvalues of equilibrium sets in Bianchi $\mathcal{B}$(VII$_0$). Refer to appendix text for explanation of columns.}
	\label{tab:EigVII0R}
\end{table}
\subsection*{Eigenvalues of equilibrium sets in $\mathcal{B}$(IV)}
\begin{table}[h]
	\centering
	\resizebox{\textwidth}{!}{\begin{tabular}{lllr}
			\toprule
			\multicolumn{4}{c}{\textbf{Eigenvalues of equilibrium sets in $\mathcal{B}$(IV) in $N_-$\,-\,gauge}}\\
			\hline
$\mathcal{P}$&s&eigenvalues $\{l_1,l_2,l_3,l_4,l_5,l_6\}$&$\#$\\
\hline
JED($\beta_1,\beta_2$)&p&$\left\{0,0,3(2-\gamma),2 \left(\beta _1+1\right),-2 \sqrt{3} \beta _2,2 \left(\beta _1+\sqrt{3} \beta _2+1\right)\right\}$&2\\
JS($\beta_1,\beta_2,\Theta$)&p&$\left\{0,0,0,2 \left(\beta _1+1\right),-2 \sqrt{3} \beta _2,2 \left(\beta _1+\sqrt{3} \beta _2+1\right)\right\}$&3\\
flat FLRW &p&$\left\{-\frac{3}{2}(2-\gamma),-\frac{3}{2}(2-\gamma),-\frac{3}{2}(2-\gamma),-\frac{3}{2}(2-\gamma),\frac{3}{2} (\gamma -\frac{2}{3}),\frac{3}{2} (\gamma -\frac{2}{3}),\frac{3}{2} (\gamma -\frac{2}{3})\right\}$&\\
open FLRW &p&$\{0,0,0,-2,-2,-2,-2\}$ &0\\
PW$(\beta_1,\nu^2)$&p&$\left\{0,0,-2 \left(\beta _1+1\right),-4 \beta _1-3 \gamma +2,-2\left(1+\beta_1\,\pm\,i\,2 \sqrt{3}\nu\,\right)\right\}$&2 \\
M&p&$\{0,0,-3(\gamma-\frac{2}{3}),-2,-2,-2\}$&2\\
W($\lambda$)&p&$\left\{0,0,-\frac{3}{2}\left(2-\gamma\right),-\frac{3}{2}\left(2-\gamma\right),-\frac{3}{4} (2-\gamma ) \left(1\pm\sqrt{6 \gamma  \left(\kappa ^2-1\right)-4 \kappa
   ^2+5}\right)\right\}$ &2\\
CS&p&$\left\{-\frac{3}{2}(2-\gamma),\frac{9}{8} (\gamma -\frac{2}{3}),\frac{9}{8} (\gamma -\frac{2}{3}),\frac{21}{8} (\gamma -\frac{10}{7}),-\frac{3}{512} \left(-3 \gamma 
   (\gamma  (3 \gamma -10)+52)+264\,\pm\,\sqrt{U}\right)\right\}$\\ \hline 
	\end{tabular}}
	\caption{Table of eigenvalues of equilibrium sets in $\mathcal{B}$(IV) calculated in $N_-$\,-\,gauge. The function $U$ is specified in the appendix text. Also, refer to the appendix introduction for an explanation of columns.}
	\label{tab:EigIVNm}
\end{table}
In table \ref{tab:EigIVNm}, the function $U$ is such that
\begin{equation}
U(\gamma)=(\gamma -2) \left(\gamma  \left(3 \gamma  \left(3 \gamma  \left(9 \gamma ^2-42 \gamma -184\right)-3760\right)+132304\right)-104480\right).
\end{equation}

\subsection*{Eigenvalues of equilibrium sets in $\mathcal{B}$(II)}
\begin{table}[H]
	\centering
	\resizebox{\textwidth}{!}{\begin{tabular}{llllr}
			\toprule
			\multicolumn{5}{c}{\textbf{Eigenvalues of equilibrium sets in $\mathcal{B}$(II) in $N_-$\,-\,gauge}}\\
			\hline
$\mathcal{P}$&s&eigenvalues $\{l_1,l_2,l_3,l_4,l_5,l_6\}$&$\#$&Parametrized\\
\hline
flat FLRW& e&$\left\{-\frac{3}{2} (2-\gamma ),-\frac{3}{2} (2-\gamma ),-\frac{3}{2} (2-\gamma ),-\frac{3}{2} (2-\gamma ),-\frac{3}{2} (2-\gamma),\frac{3}{2} (\gamma -\frac{2}{3}),\frac{3}{2} (\gamma- \frac{2}{3})\right\}$&&no\\
 JED($\beta_1,\beta_2$)&p& $\left\{0,0,3 (2-\gamma ),-2 \sqrt{3} \beta _2,\sqrt{3} \beta _2-3 \beta _1,2 \left(\beta _1+\sqrt{3} \beta _2+1\right)\right\}$&2&no\\
JS($\beta_1,\beta_2,\Theta$)&p&$\left\{0,0,0,-2 \sqrt{3} \beta _2,\sqrt{3} \beta _2-3 \beta _1,2 \left(\beta _1+\sqrt{3} \beta _2+1\right)\right\}$&3&\\
CS&p&$\left\{-\frac{3}{2}(2-\gamma ),-\frac{3}{2}(2-\gamma ),-\frac{3}{2}(2-\gamma ),\frac{9}{8} (\gamma -\frac{2}{3}),-\frac{3}{8} \left(2(2-\gamma)\pm\sqrt{F}\right)
\right\}$\\
W&p&$\left\{0,-\frac{3}{2}(2-\gamma),-\frac{3}{2}(2-\gamma),\frac{15}{4} (\gamma -\frac{6}{5}),-\frac{3}{4} (2-\gamma)\left(1 \pm\sqrt{ (5-6 \gamma )}\right)
\right\}$&0\\
R&p&$\left\{3 (3 \gamma -4),-\frac{3}{2} (5 \gamma -6),-\frac{24}{32} \left(2-\gamma \pm\sqrt{G\pm 2 \sqrt{U}}\right)
\right\}$&&\\
E&p&$\left\{0,-1,-1,3(\frac{4}{3}- \gamma),-\frac{1}{2} \left(1\,\pm\,i \sqrt{23}\right)\right\}$&0&\\ \hline 
	\end{tabular}}
	\caption{Table of eigenvalues of equilibrium sets in $\mathcal{B}$(II). The functions $U$ and $G$ are such that all the four eigenvalues of Rope not written out explicitly are negative on the $\gamma$ domain of Rope. The last column refers to whether or not the parametrization \eref{param} has been used. Refer to appendix text for further explanation of columns.}
	\label{tab:StabISS1}
\end{table}
In the table \ref{tab:StabISS1}, the functions $F,\,G$ and $U$ are such that
\begin{eqnarray}
F=2 (2-\gamma) \left(3 \gamma ^2-22 \gamma +16\right),\\
G=(2-\gamma ) \left(18 \gamma ^2-97 \gamma +90\right),\\
U=(\gamma -2)^2 \left(81 \gamma ^4+216 \gamma ^3-840 \gamma ^2+544 \gamma +16\right).
\end{eqnarray}

\end{document}